\documentclass[11pt]{article} 
\usepackage{fullpage}
\usepackage{amssymb,latexsym,mathptmx,amsmath}
\usepackage{epsfig,graphicx}
\usepackage{float}
\usepackage{leqno}
\usepackage[scaled]{helvet}
\usepackage{mathrsfs}
\usepackage{url}
\usepackage{hyperref}
\usepackage{textcomp} 
\usepackage{parskip}
\usepackage{alltt}
\usepackage{verbdef}
\usepackage{xspace}
\usepackage{verbatim}
\usepackage{enumitem}
\usepackage{lipsum}
\usepackage{wrapfig}
\usepackage{pgf,pgfrcs}
\usepackage{alg}

\usepackage{adjustbox} 
\usepackage{array}
\usepackage{booktabs}
\usepackage{multirow}
\usepackage{pifont}

\setlist[itemize]{leftmargin=*}
\newcounter{nr}

{\begin{list}{}{\setlength{\rightmargin}{0em}}\item[]}{\end{list}}
\newcommand{\qed}{\Box}
\newenvironment{proof}{{\it Proof:}}{\hfill $\qed$ \\[2mm]}
\newenvironment{prooff}{{\it Proof:}}{}

\newtheorem{thm}{Theorem}
\newtheorem{example}{Example}[section]
\newtheorem{definition}[example]{Definition}
\newtheorem{lemma}[example]{Lemma}
\newtheorem{proposition}[example]{Proposition}

\newcommand{\true}{\textit{true}}
\newcommand{\false}{\textit{false}}

\newcommand{\mkset}[1]{\{{#1}\}}

\newcommand{\upd}[3]{{#1}[{#2} \mapsto {#3}]}

\newcommand{\stotwo}[4]{\{{#1} \mapsto {#2},\ {#3} \mapsto {#4}\}}

\newcommand{\dom}[1]{\mathit{dom}(#1)}
\newcommand{\Zz}{\mathbb{Z}}

\newcommand{\eqdef}{\mathrel{\stackrel{\mathit{def}}{=}}}

\newcommand{\eg}{{\it e.g.}}

\newcommand{\etal}{{\it et al.}}

\newcommand{\labv}[1]{\mathit{Lab}(#1)}
\newcommand{\skipv}{\mathbf{Skip}}
\newcommand{\assignv}[2]{{#1} := {#2}}
\newcommand{\rassignv}[2]{{#1} := \mathbf{Random}(#2)}

\newcommand{\observev}[1]{\mathbf{Observe}(#1)}

\newcommand{\retv}[1]{\mathbf{Return}(#1)}

\newcommand{\startv}{\mathtt{Start}}
\newcommand{\observevv}{\mathbf{Observe}}
\newcommand{\finalv}{\mathtt{End}}

\newcommand{\skipc}{\mathbf{skip}}
\newcommand{\seqc}[2]{{#1}\; ;\; {#2}}
\newcommand{\ifc}[3]{\mathbf{if}\ {#1}\ \mathbf{then}\ {#2}\ \mathbf{else}\ {#3}}
\newcommand{\whilec}[2]{\mathbf{while}\ {#1}\ \mathbf{do}\ {#2}}
\newcommand{\whileck}[3]{\mathbf{while}\ {#1}\ \mathbf{do}_{{#3}}\ {#2}}
\newcommand{\progc}[2]{{#1}\ \mathbf{return}({#2})}

\newcommand{\unodes}{\mathcal{V}}

\newcommand{\PD}{\ensuremath{\mathsf{PD}}}
\newcommand{\LAP}[2]{\ensuremath{\mathsf{LAP}({#1},{#2})}}
\newcommand{\LAPf}{\ensuremath{\mathsf{LAP}}}
\newcommand{\fppd}[1]{1PPD(#1)}

\newcommand{\chain}[2]{\{{#1} \mid {#2}\}}
\newcommand{\limit}[2]{\mathit{lim}_{#1 \rightarrow \infty}\,{#2}}
\newcommand{\contarrow}{\rightarrow_{c}}
\newcommand{\fixed}[1]{\mathit{fix}({#1})}

\newcommand{\uvar}{{\cal U}}

\newcommand{\fulls}{\mathcal{F}}

\newcommand{\Dist}{\ensuremath{{\cal D}}}
\newcommand{\sumd}[1]{\sum {#1}}
\newcommand{\Dijr}[3]{\ensuremath{D_{{#1},{#2}}^{{#3}}}}

\newcommand{\selectf}[1]{\ensuremath{\mathsf{select}_{{#1}}}}
\newcommand{\assignf}[2]{\ensuremath{\mathsf{assign}_{{#1}:={#2}}}}
\newcommand{\rassignf}[2]{\ensuremath{\mathsf{rassign}_{{#1}:={#2}}}}
\newcommand{\semee}{[\![\;]\!]}
\newcommand{\seme}[1]{[\![{#1}]\!]}
\newcommand{\semc}[1]{[\![{#1}]\!]}
\newcommand{\HH}[1]{{\cal H}_{{#1}}}
\newcommand{\HHH}{{\cal H}}
\newcommand{\hpd}[3]{{#1}^{({#2},{#3})}}
\newcommand{\hh}[2]{\hpd{h}{{#1}}{{#2}}}

\newcommand{\dt}[1]{\index{#1}\textbf{\emph{#1}}} 

\newcommand{\trl}{{\cal T}}

\def\today{\ifcase\month\or
  January\or February\or March\or April\or May\or June\or
  July\or August\or September\or October\or November\or December\fi
  \space\number\day, \number\year}

\begin{document}

\title{A Semantics for Probabilistic Control-Flow Graphs}

 \author{
\begin{tabular}{cc}
Torben Amtoft & Anindya Banerjee \\
Kansas State University & IMDEA Software Institute \\
Manhattan, KS, USA & Madrid, Spain \\
\texttt{tamtoft@ksu.edu} & \texttt{anindya.banerjee@imdea.org}
\end{tabular}
}


\maketitle

\begin{abstract}
This article develops a novel operational semantics for probabilistic control-flow graphs (pCFGs) of probabilistic imperative programs with random assignment and ``observe'' (or conditioning) statements. The semantics transforms probability distributions (on stores) as control moves from one node to another in pCFGs. We relate this semantics to a standard, expectation-transforming, denotational semantics of structured probabilistic imperative programs, by translating structured programs into (unstructured) pCFGs, and proving adequacy of the translation. This shows that the operational semantics can be used without loss of information, and is faithful to the ``intended'' semantics and hence can be used to reason about, for example, the correctness of transformations (as we do in a companion article).

\end{abstract}

\section{Introduction}
We consider structured, probabilistic imperative programs that contain random assignments and ``observe'' (or conditioning) statements, in addition to usual control structures (sequences, branches, and loops). 
The recent works of Gordon~\etal~\cite{Gor+etal:ICSE-2014} and Hur~\etal~\cite{Hur+etal:PLDI-2014} present denotational semantics of such structured programs. This semantics transforms expectation functions: given a statement $S$ and an expectation function $F'$ that gives the expected return value for a store \emph{after} execution of $S$, the semantics yields an expectation function $F$ that takes a store \emph{before} execution of $S$, and gives its expected return value.

In a companion article~\cite{Amt+Ban:ProbSlicing-2017} we consider (unstructured) probabilistic control flow graphs (pCFGs) of imperative programs: our aim there is to extend classical notions of program dependence to give a semantic foundation for the slicing of probabilistic programs represented as pCFGs. To this end, we develop a novel operational semantics of pCFGs that transforms
a probability distribution at a node (say $v$) into a probability distribution at node (say $v'$), so as to model what happens when
control moves from $v$ to $v'$ in the pCFG.

A natural question is: how are the two semantics related? To wit, consider the translation of a structured probabilistic program into a 
pCFG. Are the semantics of the source and target programs adequately related? This article answers the question in the affirmative.

The implications of adequacy are at least twofold. First it shows that no information is lost by using the operational semantics: the expectation of the initial store of a structured program can always be retrieved from the probability distribution computed by the operational semantics of the program's CFG. Secondly, for deterministic programs the two semantics coincide: if the program terminates (that is, the operational semantics computes a probability distribution that is 1 for the final store) then the expected return value of the initial store of the structured program coincides with the actual return value in the final store; if the program loops (that is, the operational semantics computes the probability distribution 0), then the expected return value of the initial store of the structured program is also 0.

The authors presented a preliminary version of the operational semantics
in \cite{Amt+Ban:FoSSaCS-2016} where it was used to reason 
about the correctness of slicing probabilistic programs. 
However, in that work, the semantics was not explicitly expressed 
as a fixed point of a functional, and no comparison was made to the semantics of a structured language.

\section{The Probabilistic Control-Flow Graph Language}
\label{sec:cfg-lang}
In this section we define our language for
expressing probabilistic programs using pCFGs.
First we present the syntax
(Section~\ref{sec:CFG}) and next the semantics
(Section~\ref{sec:sem}). 

As illustrating examples,
we shall use the pCFGs depicted in Figure~\ref{fig:ex12}.
Both make use of
a random distribution $\psi_4$ 
that distributes evenly over $\mkset{0,1,2,3}$ in that
$\psi_4(0) = \psi_4(1) = \psi_4(2) = \psi_4(3) = 
\frac{1}{4}$
whereas $\psi_4(i) = 0$ for $i \notin \{0,1,2,3\}$. 

The pCFG $G_1$ on the left uses $\psi_4$
to randomly assign values to first $x$ and next $y$; 
if the sum of those values is at least 5 
then $x$ is returned, otherwise the execution is ignored.
The only possible return values are 2 and 3, with 3 twice as likely
as 2, since an $x$-value of 3 will be accepted when $y$ is 2 or 3,
whereas an $x$-value of 2 will be accepted only when $y$ is 3.

The pCFG $G_2$ on the right uses $\psi_4$ to randomly assign
a value to $x$, whereas $y$ is assigned 0; if $x$ is 0 or 1 then
$x$ is returned but otherwise the assignment in node 5 is repeated
as long as $y$ is less than 3. Depending on what $A$ is,
this may go on forever (if $A$ is say 1) or
terminate after a bounded number of steps (if $A$ is say $y + 1$);
if $A$ is a random expression that may or may not assume a value $\geq 3$
then the cycle between nodes 4 and 5 may iterate an unbounded number of time,
but will terminate with probability 1.

\begin{figure}
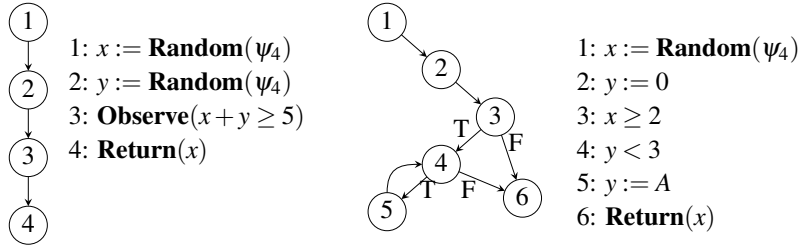

\begin{pgfpicture}{0mm}{5mm}{130mm}{35mm}
\begin{pgfmagnify}{0.9}{0.9}

\pgfnodecircle{v1}[stroke]{\pgfxy(2.5,3.5)}{8pt}

\pgfputat{\pgfxy(2.5,3.4)}{\pgfbox[center,base]{1}}

\pgfnodecircle{v2}[stroke]{\pgfxy(2.5,2.5)}{8pt}

\pgfputat{\pgfxy(2.5,2.4)}{\pgfbox[center,base]{2}}

\pgfnodecircle{v3}[stroke]{\pgfxy(2.5,1.5)}{8pt}

\pgfputat{\pgfxy(2.5,1.4)}{\pgfbox[center,base]{3}}

\pgfnodecircle{v4}[stroke]{\pgfxy(2.5,0.5)}{8pt}

\pgfputat{\pgfxy(2.5,0.4)}{\pgfbox[center,base]{4}}

\pgfsetarrowsend{stealth}
\pgfsetendarrow{\pgfarrowtriangle{4pt}}

\pgfnodeconnline{v1}{v2}

\pgfnodeconnline{v2}{v3}

\pgfnodeconnline{v3}{v4}

\pgfnodecircle{w1}[stroke]{\pgfxy(7.8,3.5)}{8pt}

\pgfputat{\pgfxy(7.8,3.4)}{\pgfbox[center,base]{1}}

\pgfnodecircle{w2}[stroke]{\pgfxy(8.6,2.8)}{8pt}

\pgfputat{\pgfxy(8.6,2.7)}{\pgfbox[center,base]{2}}

\pgfnodecircle{w3}[stroke]{\pgfxy(9.4,2.1)}{8pt}

\pgfputat{\pgfxy(9.4,2.0)}{\pgfbox[center,base]{3}}

\pgfnodecircle{w4}[stroke]{\pgfxy(8.6,1.4)}{8pt}

\pgfputat{\pgfxy(8.6,1.3)}{\pgfbox[center,base]{4}}

\pgfnodecircle{w5}[stroke]{\pgfxy(7.8,0.7)}{8pt}

\pgfputat{\pgfxy(7.8,0.6)}{\pgfbox[center,base]{5}}

\pgfnodecircle{w6}[stroke]{\pgfxy(9.8,0.9)}{8pt}

\pgfputat{\pgfxy(9.8,0.8)}{\pgfbox[center,base]{6}}

\pgfsetarrowsend{stealth}
\pgfsetendarrow{\pgfarrowtriangle{4pt}}

\pgfnodeconnline{w1}{w2}

\pgfnodeconnline{w2}{w3}

\pgfnodeconnline{w3}{w4}

\pgfputat{\pgfxy(8.9,1.8)}{\pgfbox[center,base]{T}}

\pgfputat{\pgfxy(9.7,1.6)}{\pgfbox[center,base]{F}}

\pgfnodeconnline{w3}{w6}

\pgfnodeconnline{w4}{w5}

\pgfputat{\pgfxy(8.4,0.9)}{\pgfbox[center,base]{T}}

\pgfputat{\pgfxy(9.0,0.9)}{\pgfbox[center,base]{F}}

\pgfnodeconnline{w4}{w6}

\pgfnodeconncurve{w5}{w4}{90}{180}{10pt}{10pt}

\pgfputat{\pgfxy(3.1,3.0)}{\pgfbox[left,base]{1: $\rassignv{x}{\psi_4}$}}
\pgfputat{\pgfxy(3.1,2.5)}{\pgfbox[left,base]{2: $\rassignv{y}{\psi_4}$}}
\pgfputat{\pgfxy(3.1,2.0)}{\pgfbox[left,base]{3: $\observev{x+y \geq 5}$}}
\pgfputat{\pgfxy(3.1,1.5)}{\pgfbox[left,base]{4: $\retv{x}$}}

\pgfputat{\pgfxy(10.6,3.0)}{\pgfbox[left,base]{1: $\rassignv{x}{\psi_4}$}}
\pgfputat{\pgfxy(10.6,2.5)}{\pgfbox[left,base]{2: $\assignv{y}{0}$}}
\pgfputat{\pgfxy(10.6,2.0)}{\pgfbox[left,base]{3: $x \geq 2$}}
\pgfputat{\pgfxy(10.6,1.5)}{\pgfbox[left,base]{4: $y < 3$}}
\pgfputat{\pgfxy(10.6,1.0)}{\pgfbox[left,base]{5: $\assignv{y}{A}$}}
\pgfputat{\pgfxy(10.6,0.5)}{\pgfbox[left,base]{6: $\retv{x}$}}

\end{pgfmagnify}
\end{pgfpicture}

\caption{\label{fig:ex12} 
A pCFG $G_1$ with an $\observevv$ node (left);
a pCFG $G_2$ with a cycle (right).}
\end{figure}

\subsection{Syntax}
\label{sec:CFG}

This section describes the kind of control flow graphs that we consider;
special emphasis is on the notion of postdomination.
We use Figure~\ref{fig:ex12} to motivate our approach.

A node $v \in \unodes$
can be labeled (the label is called $\labv{v}$) with an assignment
$\assignv{x}{E}$ ($x$ a program variable
and $E$ an arithmetic expression),
with a random assignment $\rassignv{x}{\psi}$
with $\psi$ a probability distribution
(which we assume
contains no program variables though it would be straightforward
to allow it as in Hur~\etal~\cite{Hur+etal:PLDI-2014}),
with a \emph{conditioning} $\observev{B}$ ($B$ is a boolean expression),
or (though not part of these examples) with $\skipv$;
a node of the abovementioned kinds has exactly one outgoing edge.
Also, there are branching nodes with \emph{two}
outgoing edges. (If $v$ has an outgoing edge to $v'$ we say that
$v'$ is a successor of $v$; a branching node has a $\true$-successor
and a $\false$-successor.) 
Finally, there is
a special node $\startv$ (which 
is numbered 1 in the examples)
from which there is a path to all other nodes,
and a unique $\finalv$ node with no outgoing edges
(if it has a label
it will be of the form $\retv{x}$)
to which there is a path from all other nodes.

\begin{definition}[Deterministic pCFG]
\label{def:determCFG}
We say that a pCFG is deterministic
if it has no $\observevv$ nodes or random assignments.
\end{definition}

\subsubsection{Postdomination}
Whereas the notion of domination has early been used
to reason about the structure of general 
control-flow graphs~\cite{Aho:86:Dragon},
the dual notion of postdomination has been a crucial concept in
standard work on dependences and slicing
(such as~\cite{Pod+Cla:TSE-1990,Bal+Hor:AAD-1993}),
and will also play a key part in our development.

\begin{definition}[Postdomination]
We say that $v_1$ \dt{postdominates} $v$, also written
$(v,v_1) \in \PD$, if $v_1$ occurs
on all paths from $v$ to $\finalv$;
if also $v_1 \neq v$,
$v_1$ is a \dt{proper postdominator} of $v$.
\end{definition}
It is easy to see that
the relation ``postdominates'' is reflexive, transitive, and antisymmetric.
\begin{lemma}
\label{lem:prec-ordering}
For given $v$, let $\prec$ be an ordering among
proper postdominators of $v$,
by stipulating that $v_1 \prec v_2$ iff in all acyclic paths from $v$
to $\finalv$, $v_1$ occurs strictly before $v_2$.
Then $\prec$ is transitive, antisymmetric, and total.
Also, if $v_1 \prec v_2$ then 
for \emph{all} paths from $v$ to $\finalv$ it is the case that
the first occurrence of $v_1$ is before the first occurrence of $v_2$.
\end{lemma}
\begin{proof}
The first two properties are obvious.

We next show that $\prec$ is total.
Assume, to get a contradiction, that there exists an 
acyclic path $\pi_1$ from $v$ to $\finalv$
that contains $v_1$ strictly before $v_2$,
and also an
acyclic path $\pi_2$ from $v$ to $\finalv$
that contains $v_2$ strictly before $v_1$.
But then the concatenation of the prefix of $\pi_1$ that ends with $v_1$,
and the suffix of $\pi_2$ that starts with $v_1$, is a path
from $v$ to $\finalv$
that avoids $v_2$, yielding a contradiction
as $v_2$ postdominates $v$.

Finally, assume that $v_1 \prec v_2$, and that $\pi$ is a path
from $v$ to $\finalv$; to get a contradiction, assume that
there is a prefix $\pi_1$ of $\pi$ that ends with $v_2$ but does 
not contain $v_1$. Since there exists an acyclic path
from $v$ to $\finalv$, we infer from $v_1 \prec v_2$ that
there is an acyclic path $\pi_2$ from $v_2$ that does not contain $v_1$.
But the concatenation of $\pi_1$ and $\pi_2$ is a path from
$v$ to $\finalv$ that does not contain $v_1$,
which contradicts $v_1$ being a proper postdominator of $v$.
\end{proof}
We say that $v_1$ is the \dt{first proper postdominator} of $v$
if whenever $v_2$ is another proper postdominator of $v$
then all paths from $v$ to $v_2$ contain $v_1$.
\begin{lemma}
For any $v$ with $v \neq \finalv$,
there is a unique first proper postdominator of $v$.
\end{lemma}
\begin{proof}
This follows from Lemma~\ref{lem:prec-ordering}
since the set of proper postdominators is
non-empty (will at least contain $\finalv$).
\end{proof}
\begin{definition}
For $v \neq \finalv$, we
write $\fppd{v}$ for the unique first proper postdominator of $v$.
\end{definition}
In Figure~\ref{fig:ex12}(right), $\fppd{1} = 2$
(while also nodes 3 and 6 are proper postdominators of 1)
and $\fppd{3} = 6$.

\subsubsection{Nodes that Induce Cycles}
In order to inductively define functions on pairs 
in $\PD$, we need to define a measure on such pairs:

\begin{definition}[\LAPf]
For $(v,v') \in \PD$, we define $\LAP{v}{v'}$ as the maximum length
of an acyclic path from $v$ to $v'$. 
(The length of a path is the number of edges.)
\end{definition}
Thus $\LAP{v}{v} = 0$ for all nodes $v$. 
As expected, we have:
\begin{lemma}
\label{lem:LAP-add}
If $(v,v_1) \in \PD$ and $(v_1,v_2) \in \PD$ (and thus
$(v,v_2) \in \PD$) then $\LAP{v}{v_2} = \LAP{v}{v_1} + \LAP{v_1}{v_2}$.
\end{lemma}
\begin{proof}
If $v = v_1$ or $v_1 = v_2$, the claim is obvious;
we can thus assume that $v_1$ and $v_2$ are proper postdominators of $v$
and that $v_1 \prec v_2$ in the total ordering among those so
by Lemma~\ref{lem:prec-ordering} we see that $v_1$
will occur before $v_2$ in all paths from $v$ to $\finalv$.

First consider an acyclic path $\pi$ from $v$ to $v_2$. 
We have argued that $\pi$ will contain $v_1$,
and hence $\pi$ is the concatenation of an acyclic path from
$v$ to $v_1$, thus of length $\leq \LAP{v}{v_1}$,
and an acyclic path from $v_1$ to $v_2$, thus of length $\leq \LAP{v_1}{v_2}$.
Thus the length of $\pi$ is at most $\LAP{v}{v_1} + \LAP{v_1}{v_2}$;
as $\pi$ was an arbitrary acyclic path from $v$ to $v_2$, 
this shows ``$\leq$''. 

To show ``$\geq$'', let $\pi_1$ be an acyclic path from
$v$ to $v_1$ of length $\LAP{v}{v_1}$, and $\pi_2$ be an acyclic path from
$v_1$ to $v_2$ of length $\LAP{v_1}{v_2}$. 
Let $\pi$ be the concatenation
of $\pi_1$ and $\pi_2$;
$\pi$ is an acyclic path from $v$ to $v_2$
since if $v' \neq v_1$ occurs in both
paths then there is a path from $v$ to $v_2$ that avoids $v_1$
which is a contradiction. 
As $\pi$ is of length 
$\LAP{v}{v_1} + \LAP{v_1}{v_2}$, this shows ``$\geq$''.
\end{proof}

To reason about cycles, 
it is useful to pinpoint the kind of nodes that cause cycles:
\begin{definition}[Cycle-inducing]
\label{def:cycle-induce}
A node $v$ is \emph{cycle-inducing} if with $v' = \fppd{v}$ there exists a successor $v_i$
of $v$ such that $\LAP{v_i}{v'} \geq \LAP{v}{v'}$.
\end{definition}
Note that if $v$ is cycle-inducing then $v$ must be a branching node
(since if $v$ has only one successor then that successor is $v'$).
\begin{example}
\label{ex:cycle-induce}
In Figure~\ref{fig:ex12}(right), there are two branching nodes,
3 and 4, both having node 6 as their first proper postdominator.
Node 4 is cycle-inducing, since
5 is a successor of 4 with
$\LAP{5}{6} = 2 > 1 = \LAP{4}{6}$. 
On the other hand, node 3 is not cycle-inducing, since
$\LAP{3}{6} = 2$ which is strictly greater
than $\LAP{4}{6}$ ($=1$) and $\LAP{6}{6}$ ($=0$).
\end{example}
\begin{lemma}
\label{lem:LAPcycle}
If $v$ is cycle-inducing then there exists a cycle that contains $v$
but not $\fppd{v}$.
\end{lemma}
\begin{proof}
With $v' = \fppd{v}$, by assumption there exists 
a successor $v_i$ of $v$ such that $\LAP{v_i}{v'} \geq \LAP{v}{v'}$;
observe that $v'$ is a postdominator of $v_i$.
Let $\pi$ be an acyclic path from $v_i$ to $v'$ with length
$\LAP{v_i}{v'}$;
then the path $v \pi$ is a path from $v$ to $v'$ that is longer
than $\LAP{v_i}{v'}$, and thus also longer
than $\LAP{v}{v'}$. This shows that $v \pi$ cannot be acyclic;
hence $v \in \pi$ and thus $v \pi$ contains a cycle involving $v$ but not $v'$.
\end{proof}
\begin{lemma}
\label{lem:cycle1inducing}
All cycles will contain at least one node which is cycle-inducing.
\end{lemma}
\begin{proof}
Let a cycle $\pi$ be given. For each $v \in \pi$, define 
$f(v)$ as $\LAP{v}{\finalv}$. For a node $v$ that 
has only one successor, $v_1$,
it holds that $f(v_1) < f(v)$ (since by Lemma~\ref{lem:LAP-add} we 
have $f(v) = \LAP{v}{v_1} + f(v_1) = 1 + f(v_1)$).
Thus $\pi$ must contain a branching node $v_0$ with a successor $v_i$
such that $f(v_i) \geq f(v_0)$.
But with $v' = \fppd{v_0}$ we then have (by Lemma~\ref{lem:LAP-add})
\[
\LAP{v_i}{v'} = f(v_i) - f(v')
\geq f(v_0) - f(v')
= \LAP{v_0}{v'}
\]
which shows that $v_0$ is cycle inducing.
\end{proof}

\subsection{Semantics}
\label{sec:sem}
In this section we shall define the meaning of the pCFGs introduced
in Section~\ref{sec:CFG}, in terms of an operational semantics 
that manipulates \emph{distributions} which
as described in Section~\ref{subsec:sem-dist}
assign probabilities to stores.
In Section~\ref{subsec:sem-domain} we shall explore the structure
of distributions, and operators on such, using
basic concepts from domain theory as summarized in Section~\ref{subsec:domain}.
In Section~\ref{subsec:sem-one} we shall define the semantics
of traversing one edge of a pCFG; based on that,
in Section~\ref{subsec:toplevel} we shall present a functional
the fixed point of which provides the meaning of a pCFG. 
In Section~\ref{subsec:sem-ex} we shall illustrate the semantics
on the pCFGs from Figure~\ref{fig:ex12}.

\subsubsection{Stores and Distributions}
\label{subsec:sem-dist}
Let $\uvar$ be the universe of variables.
A store $s \in \fulls$ is a mapping from $\uvar$ to $\Zz$.
We write $\upd{s}{x}{z}$ for the store $s'$ that is like $s$
except $s'(x) = z$.
We assume that there is a function $\semee$
such that $\seme{E}s$ is the integer result of evaluating $E$ in store $s$
and $\seme{B}s$ is the boolean result of evaluating $B$ in store $s$
(the free variables of $E$ and $B$ must be in $\dom{s}$).

A distribution $D \in \Dist$ 
is a mapping from $\fulls$ to non-negative reals.
We shall often expect that $D$ is \emph{bounded}, that is
$\sumd{D} \leq 1$
where $\sumd{D}$ is a shorthand for $\sum_{s \in \fulls}D(s)$.
Thanks to our assumption that values are integers, and since $\uvar$
can be assumed finite,
$\fulls$ is a countable set and thus $\sumd{D}$ is well-defined
even without measure theory.
We define $D_1 + D_2$ by stipulating
$(D_1 + D_2)(s) = D_1(s) + D_2(s)$
(if $\sumd{D_1} + \sumd{D_2} \leq 1$ then $D_1 + D_2$ is bounded),
and for $c \geq 0$ we define $cD$ by stipulating 
$(cD)(s) = cD(s)$ (if $D$ is bounded and $c \leq 1$ then $cD$ is bounded);
we write $D = 0$ when $D(s) = 0$ for all $s$.
We say that $D$ is \emph{concentrated} if there exists $s_0 \in \fulls$
such that $D(s) = 0$ for all $s \in \fulls$ with $s \neq s_0$;
for that $s_0$, we say that $D$ is concentrated on $s_0$.
(Thus the distribution $0$ is concentrated on everything.)

\subsubsection{Basic Domain Theory}
\label{subsec:domain}
To prepare for our development, we shall
now recall some key notions of domain theory,
as presented in, \eg,
\cite{Schmidt:DenSemantics,Winskel:semantics}.

Let $(X, \preceq)$ be a partially ordered set. 
A \emph{chain} $\chain{x_k}{k}$ in $X$ is a mapping from the natural numbers into $X$
such that if $i < j$ then $x_i \preceq x_j$. The chain 
has a \emph{least upper bound} $x \in X$, if $x_k \preceq x$ for all $k$,
and also if $x_k \preceq y$ for all $k$ then $x \preceq y$. (The least upper bound is also called \emph{limit} and written $\limit{k}{x_k}$.) 
Finally, say that $(X, \preceq)$ is a \emph{complete partial order (cpo)} if every chain $\chain{x_k}{k}$ in $X$ has a least upper bound.
We say that a cpo is a \emph{pointed} cpo if there exists a least
element (also called \emph{bottom}), 
that is an element $\bot$ such that $\bot \preceq x$ 
for all $x \in X$.

A function $f$ from a cpo $X$ to a cpo $Y$ is \emph{continuous}
if for each chain $\chain{x_k}{k}$ in $X$ the following holds:
$\chain{f(x_k)}{k}$ is a chain in $Y$,
and $\limit{k}{f(x_k)} = f(\limit{k}{x_k})$.
We let $X \contarrow Y$ denote the set of 
continuous functions from $X$ to $Y$.
A continuous function $f$ is also monotone, that is $f(x_1) \preceq f(x_2)$
when $x_1 \preceq x_2$
(for then $x_1,x_2,x_2,x_2....$ is a chain
and by continuity thus
$f(x_2)$ is the least upper bound of $f(x_1),f(x_2)$
implying $f(x_1) \preceq f(x_2)$).
For a function $f$, and for each $k \geq 0$ 
define $f^k$ by stipulating
$f^0(x) = x$, and $f^{k+1}(x) = f(f^{k}(x))$ for $k \geq 0$.
For a function $f$ from a cpo $Y$ into itself,
a \emph{fixed point} of $f$ is an element $y \in Y$ such that
$f(y) = y$.
\begin{lemma}
\label{lem:cont-cpo-fix}
Let $f$ be a continuous function on
a pointed cpo $X$. 
Then $\chain{f^k(\bot)}{k}$ is a chain,
and $\limit{k}{f^k(\bot)}$ is the least fixed point of $f$.
\end{lemma}
\begin{proof}
From $\bot \preceq f(\bot)$ we use monotonicity of $f$ to infer that
$f^k(\bot) \preceq f^{k+1}(\bot)$ for all $k$
so $\chain{f^k(\bot)}{k}$ is indeed a chain.
With $x = \limit{k}{f^k(\bot)}$ we see by continuity of $f$ 
that $x$ is indeed a fixed point of $f$:
$f(x)= \limit{k}{f^{k+1}(\bot)} = x$.
And if $y$ is also a fixed point, we have
$\bot \preceq y$ and by monotonicity of $f$ thus
$f^k(\bot) \preceq f^{k}(y) = y$ for all $k$, from which
we infer $x \preceq y$.
\end{proof}
\begin{lemma}
\label{lem:cont-cpo}
Let $X$ and $Y$ be cpos. Then $X \contarrow Y$
is a cpo (and pointed if $Y$ is), 
with ordering and limit defined pointwise: $f_1 \preceq f_2$
iff $f_1(x) \preceq f_2(x)$ for all $x \in X$,
and $f = \limit{k}{f_k}$ if $f$ is 
such that $f(x) = \limit{k}{(f_k(x))}$ for all $x \in X$.
\end{lemma}
\begin{proof}
For a chain $\chain{k}{f_k} \in X \contarrow Y$, 
$f$ as defined in the lemma text
is clearly a least upper bound,
provided we can show
that $f$ is continuous.
But if $\chain{x_k}{k}$ is a chain in $X$ then
\begin{eqnarray*}
f(\limit{k}{x_k}) & = & \limit{m}{f_m(\limit{k}{x_k})}
 = \limit{m}{\limit{k}{f_m(x_k)}}
=  \limit{k}{\limit{m}{f_m(x_k)}}
\\ & = & \limit{k}{f(x_k)}.
\end{eqnarray*}
If $Y$ has a bottom element $\bot$ then
$\lambda x.\bot$ is the bottom element in
$X \contarrow Y$.
\end{proof}
\subsubsection{Relating to Domain Theory}
\label{subsec:sem-domain}
The set of real numbers in $[0..1]$ form a {\em pointed cpo}
with the usual ordering, as 0 is the bottom element and
the supremum operator yields the least upper bound of a chain.
Hence also the set $\Dist$ of distributions form a pointed cpo,
with ordering defined pointwise ($D_1 \preceq D_2$ iff
$D_1(s) \leq D_2(s)$ for all stores $s$), with 0 the bottom element,
and the least upper bound defined pointwise. 
\begin{lemma}
\label{lem:space-cpo}
The set $\Dist \contarrow \Dist$ is a pointed cpo,
as is the set $\PD \contarrow \Dist \contarrow \Dist$
where $\PD$ is considered a discrete cpo.
\end{lemma}
\begin{proof}
This follows from Lemma~\ref{lem:cont-cpo}.
\end{proof}
Observe that the ordering $\preceq$ on 
$\PD \contarrow \Dist \contarrow \Dist$ is determined as follows:
$h_1 \preceq h_2$ iff for all $(v,v') \in \PD$,
all $D \in \Dist$, and all $s \in \fulls$, it holds that
$h_1(v,v')(D)(s) \leq h_2(v,v')(D)(s)$.
Also, the least element $0$ is given as 
$\lambda (v_1,v_2).\lambda D.\lambda s.0$.

The following result
is often convenient;
in particular, it shows that
if each distribution in a chain $\chain{D_k}{k}$ is bounded
then also the least upper bound is a bounded distribution.
\begin{lemma}
\label{lem:sumlim-limsum}
Assume that $\chain{D_k}{k}$ 
is a chain of distributions (not necessarily bounded).
With $S$ a (countable) set of stores, we have
\[
\limit{k}{\sum_{s \in S}{D_k(s)}} =
\sum_{s \in S}{(\limit{k}{D_k})(s)}.
\]
\end{lemma}
\begin{proof}
Let $D' = \limit{k}{D_k}$.
From $D_k \leq D'$ we get that
$\sum_{s \in S}{D'(s)}$ is an upper bound for 
$\chain{\sum_{s \in S}{D_k(s)}}{k}$; as
$\limit{k}{\sum_{s \in S}{D_k(s)}}$ is the least upper bound,
we get
\[
\limit{k}{\sum_{s \in S}{D_k(s)}} \leq \sum_{s \in S}{D'(s)}.
\]
To establish that equality holds, we shall assume
$\limit{k}{\sum_{s \in S}{D_k(s)}} < \sum_{s \in S}{D'(s)}$
so as to get a contradiction.
Then there exists $\epsilon > 0$ such that
$\limit{k}{\sum_{s \in S}{D_k(s)}} + \epsilon < \sum_{s \in S}{D'(s)}$.
We infer that there exists a finite set $S_0$ with $S_0 \subseteq S$
such that 
$\limit{k}{\sum_{s \in S}{D_k(s)}} + \epsilon < \sum_{s \in S_0} D'(s)$.
For each $s \in S_0$ there exists $K_s$ such that
$D_k(s) > D'(s) - \epsilon/|S_0|$ for $k \geq K_s$,
and thus there exists $K$ (the maximum element of the finite set $\{K_s \mid s \in S_0\}$)
such that for each $s \in S_0$, and each $k \geq K$,
$D_k(s) + \epsilon/|S_0| > D'(s)$.
But then we get
\begin{eqnarray*}
\sum_{s \in S_0} D'(s) & < & \sum_{s \in S_0}(D_K(s) + \epsilon/|S_0|)
= \sum_{s \in S_0}D_K(s) + \epsilon
\leq \limit{k}{\sum_{s \in S}{D_k(s)}} + \epsilon
\\ & < & \sum_{s \in S_0} D'(s).
\end{eqnarray*}
which yields the desired contradiction.
\end{proof}
When developing the semantics,
we shall define a number of functions with functionality
$\Dist \rightarrow \Dist$. Each such function $f$ typically 
has a number of useful properties, such as being
\begin{itemize}
\item
\emph{continuous} (and hence monotone), as defined above;
\item \emph{additive}:
$f(D_1 + D_2) = f(D_1) + f(D_2)$ for all distribution $D_1$, $D_2$
(this reflects that a distribution is not more than the sum of its components);
\item
\emph{multiplicative}: $f(cD) = cf(D)$ for all distributions
$D$ and all real $c \geq 0$;
\item
\emph{non-increasing}:
$\sumd{f(D)} \leq \sumd{D}$ for all
distributions $D$ (this reflects that distribution may disappear,
due to to $\observevv$ nodes or infinite loops, but
cannot be created \emph{ex nihilo}).
\end{itemize}
In addition, some functions will even be
\begin{itemize}
\item
\emph{sum-preserving}: 
$\sumd{f(D)} = \sumd{D}$ for all
distributions $D$;
\item
\emph{deterministic}:
$f(D)$ is concentrated for all concentrated distributions $D$.
\end{itemize}
The semantic function for an 
$\observevv$ node is not sum-preserving (unless the condition is always true), and neither is the
semantic function for a loop that has a non-zero probability of
non-termination;
the semantic function for a random assignment 
is not deterministic 
(unless the random distribution is concentrated on one value).

\subsubsection{One-Step Semantics}
\label{subsec:sem-one}
We now define the semantics of traveling one edge
in the pCFG. 

\paragraph{$\observevv$ and branching nodes}
For such nodes, distributions are pruned; to model that,
for a boolean expression $B$ we define $\selectf{B}$ by letting
$\selectf{B}(D) = D'$ where
\[
\begin{array}{rcll}
D'(s) & = & D(s) & \mbox{if } \seme{B}s \\[1mm]
D'(s) & = & 0 & \mbox{otherwise.}
\end{array}
\]
It is then straight-forward to establish that
\begin{lemma}
\label{lem:selectB}
For all $B$,
$\selectf{B}$ is continuous, additive, multiplicative,
non-increasing, and deterministic; also, for all $D$ we have
\[
\selectf{B}(D) + \selectf{\neg B}(D) = D.
\]
\end{lemma}

\paragraph{Assignments}
For a variable $x$ and an expression $E$,
we define $\assignf{x}{E}$ by letting
$\assignf{x}{E}(D)$ 
be a distribution $D'$ such that
for each $s' \in \fulls$, 
\[
D'(s') =
\sum_{s \in \fulls\ \mid\ s' = \upd{s}{x}{\seme{E}s}} D(s).
\]
That is, the ``new'' probability of a store $s'$ is the sum of the 
``old'' probabilities
of the stores that become like $s'$ after the assignment
(this will happen for a store $s$ iff $s' = \upd{s}{x}{\seme{E}s}$).
\begin{lemma}
\label{lem:assign-misc}
\label{lem:assign-cont}
Each $\assignf{x}{E}$ is continuous, additive, multiplicative, non-increasing, sum-preserving, and deterministic.
\end{lemma}
\begin{prooff}
Additivity and multiplicativity are trivial,
and if $D$ is concentrated on $s_0$
then $\assignf{x}{E}(D)$ is concentrated on
$\upd{s_0}{x}{\seme{E}s_0}$. We are left with two non-trivial tasks.

To prove that $\assignf{x}{E}$ is {\bf sum-preserving}, 
and hence non-increasing,
we have the calculation
\begin{eqnarray*}
\sumd{D'} & = & 
\sum_{s' \in \fulls}D'(s') 
=
\sum_{s' \in \fulls} \left( \sum_{s \in \fulls\ \mid\ s' = \upd{s}{x}{\seme{E}s}} D(s) \right)
=
\sum_{s \in \fulls,\ s' \in \fulls\ \mid\ s' = \upd{s}{x}{\seme{E}s}} D(s)
=
\sum_{s \in \fulls}D(s) \\ & = & \sumd{D}.
\end{eqnarray*}
To prove that $\assignf{x}{E}$ is {\bf continuous},
let $\chain{D_k}{k}$ be a chain; with
$D'_k = \assignf{x}{E}(D_k)$,
also $\chain{D'_k}{k}$ is a chain
since $\assignf{x}{E}$ is obviously monotone.
With $D = \limit{k}{D_k}$ and $D' = \limit{k}{D'_k}$,
our goal is to prove that $D' = \assignf{x}{E}(D)$.
But this follows since by Lemma~\ref{lem:sumlim-limsum}
for each $s'$ we have the calculation
\begin{eqnarray*}
D'(s') & = & \limit{k}{D'_k(s')} =
     \limit{k}{\sum_{s \in \fulls\ \mid\ s' = \upd{s}{x}{\seme{E}s}} D_k(s)} 
\\ & = & \sum_{s \in \fulls\ \mid\ s' = \upd{s}{x}{\seme{E}s}}\limit{k}{D_k(s)}
    = \sum_{s \in \fulls\ \mid\ s' = \upd{s}{x}{\seme{E}s}} D(s).
\end{eqnarray*}
\end{prooff}

\paragraph{Random Assignments}
For a variable $x$ and a random distribution $\psi$,
we define $\rassignf{x}{\psi}$ by letting
$\rassignf{x}{\psi}(D)$
be a distribution $D'$ such that
for each $s' \in \fulls$,
\[
D'(s') = \psi(s'(x)) \left(
\sum_{s \in \fulls\ \mid\ s' = \upd{s}{x}{s'(x)}} D(s) \right).
\]
\begin{lemma}
\label{lem:rassign-misc}
\label{lem:rassign-cont}
Each $\rassignf{x}{E}$ is continuous, additive, multiplicative, non-increasing, and sum-preserving.
\end{lemma}
\begin{prooff}
Additivity and multiplicativity are trivial.
To prove that $\rassignf{x}{E}$ is {\bf sum-preserving}, 
and hence non-increasing,
we have the calculation (using that 
$\sum_{z \in \Zz}\psi(z) = 1$)
\begin{eqnarray*}
\sumd{D'} & = & \displaystyle
\sum_{s' \in \fulls}D'(s') 
=
\sum_{s' \in \fulls} \psi(s'(x))
\left( \sum_{s \in \fulls\ \mid\ s' = \upd{s}{x}{s'(x)}} D(s) \right)
=
\sum_{s \in \fulls,\ s' \in \fulls\ \mid\ s' = \upd{s}{x}{s'(x)}}
\psi(s'(x))D(s)
\\[1mm] & = & 
\sum_{s \in \fulls,\ z \in \Zz,\ s' \in \fulls\ \mid\ s'(x) = z,\ s' = \upd{s}{x}{s'(x)}}
\psi(s'(x))D(s)
=
\sum_{s \in \fulls,\ z \in \Zz,\ s' \in \fulls\ \mid\ s' = \upd{s}{x}{z}}
\psi(z)D(s)
\\[1mm] & = & \displaystyle
\sum_{s \in \fulls, z \in \Zz}
\psi(z)D(s)
=
\sum_{s \in \fulls}D(s) \sum_{z \in \Zz}\psi(z)
=
\sum_{s \in \fulls}D(s) \cdot 1
=
\sumd{D}.
\end{eqnarray*}
To prove that $\rassignf{x}{E}$ is {\bf continuous},
let $\chain{D_k}{k}$ be a chain; with
$D'_k = \rassignf{x}{E}(D_k)$,
also $\chain{D'_k}{k}$ is a chain
since $\rassignf{x}{E}$ is obviously monotone.
With $D = \limit{k}{D_k}$ and $D' = \limit{k}{D'_k}$,
our goal is to prove that $D' = \rassignf{x}{E}(D)$.
But this follows since by Lemma~\ref{lem:sumlim-limsum}
for each $s'$ we have the calculation
\begin{eqnarray*}
D'(s') & = & \limit{k}{D'_k(s')} =
     \limit{k}{\left(\psi(s'(x)) 
    \sum_{s \in \fulls\ \mid\ s' = \upd{s}{x}{s'(x)}} D_k(s) \right)}
\\ & = & 
\psi(s'(x)) \left(\limit{k}{\sum_{s \in \fulls\ \mid\ s' = \upd{s}{x}{s'(x)}} D_k(s)}\right)
  =
   \psi(s'(x)) \left( \sum_{s \in \fulls\ \mid\ s' = \upd{s}{x}{s'(x)}} \limit{k}{D_k(s)} \right)
  \\ & = &
  \psi(s'(x)) \left( \sum_{s \in \fulls\ \mid\ s' = \upd{s}{x}{s'(x)}} D(s) \right).
\end{eqnarray*}
\end{prooff}

\subsubsection{Semantics as a Fixed-Point}
\label{subsec:toplevel}
Having expressed the semantics of a single edge,
we shall now express the semantics of a full pCFG.
Our goal is to compute ``modification functions''
to express how a distribution is modified
as ``control'' moves from $\startv$ to $\finalv$.
To accomplish this, we shall solve a more general problem:
for each $(v,v') \in \PD$,
state how a given distribution
is modified as ``control'' moves from $v$ to $v'$
along paths that may contain
multiple branches and even loops
but which do \emph{not} contain $v'$ until the end.

We would have liked to have a definition of the modification
function that is inductive in $\LAP{v}{v'}$,
but this is not possible due to cycle-inducing nodes
(cf.~Definition~\ref{def:cycle-induce}).
For such nodes, the semantics cannot be expressed by recursive calls
on the successors, but the semantics of (at least) one of
the successors will have to be provided as an {\em argument}.
This motivates that our main semantic function be a \emph{functional}
$\HHH$ 
that transforms a modification function into another modification function,
with the desired meaning being the \emph{fixed point}
(cf.~Lemma~\ref{lem:cont-cpo-fix}) of this functional.
The functional $\HHH$ 
operates on 
$\PD \contarrow \Dist \contarrow \Dist$ but we shall first define it for
$\PD \rightarrow \Dist \rightarrow \Dist$:
\begin{definition}[$\HHH$]
\label{def:HHX}
The functionality of $\HHH$ is given by 
\[
\HHH: (\PD \rightarrow \Dist \rightarrow \Dist)
 \rightarrow
   (\PD \rightarrow \Dist \rightarrow \Dist)
\]
where, given 
\[
h_0: \PD \rightarrow \Dist \rightarrow \Dist
\]
we define 
\[
h = \HHH(h_0): \PD \rightarrow \Dist \rightarrow \Dist
\]
by letting $h (v,v')$, written $\hh{v}{v'}$, be stipulated by the
following rules that are inductive in $\LAP{v}{v'}$:
\begin{enumerate}
\item
if $v' = v$ then $\hh{v}{v'}(D) = D$;
\item
\label{evalk-transitive}
otherwise, if $v' \neq v''$ with $v'' = \fppd{v}$
then 
\[
\hh{v}{v'}(D) = \hh{v''}{v'}(\hh{v}{v''}(D))
\]
(this is well-defined by Lemma~\ref{lem:LAP-add});
\item
otherwise, that is if $v' = \fppd{v}$:
\begin{enumerate}
\item
if $\labv{v} = \skipv$
then $\hh{v}{v'}(D) = D$;
\item
if $\labv{v}$ is of the form $\assignv{x}{E}$
then $\hh{v}{v'}(D) = \assignf{x}{E}(D)$;
\item
if $\labv{v}$ is of the form $\rassignv{x}{\psi}$
then $\hh{v}{v'}(D) = \rassignf{x}{\psi}(D)$;
\item
if $\labv{v}$ is of the form $\observev{B}$
then $\hh{v}{v'}(D) = \selectf{B}(D)$;
\item
\label{def:evalk:cond}
otherwise, that is if $v$ is a branching node with condition $B$,
we compute $\hh{v}{v'}$ as follows:
with $v_1$ the $\true$-successor of $v$ and $v_2$ the $\false$-successor
of $v$, let $D_1 = \selectf{B}(D)$ and $D_2 = \selectf{\neg B}(D)$;
then let $\hh{v}{v'}(D) = D'_1 + D'_2$ where for each
$i \in \{1,2\}$, $D'_i$ is given as
\begin{itemize}
\item
if $\LAP{v_i}{v'} < \LAP{v}{v'}$ then
$D'_i =\hh{v_i}{v'}(D_i)$;
\item
if $\LAP{v_i}{v'} \geq \LAP{v}{v'}$ (and thus $v$ is cycle-inducing)
then 
$D'_i = \hpd{h_0}{v_i}{v'}(D_i)$.
\end{itemize}
\end{enumerate}
\end{enumerate}
\end{definition}
\begin{lemma}
\label{lem:hcontH}
$\HHH$ maps from $\PD \contarrow (\Dist \contarrow \Dist)$ to itself
(with $\PD$ considered a discrete cpo).
\end{lemma}
\begin{proof}
We first show that if
$\hpd{h_0}{v}{v'}$ is continuous for all $(v,v') \in \PD$
then with $h = \HHH(h_0)$, then 
$\hh{v}{v'}$ is continuous for all $(v,v') \in \PD$.
But this follows by an easy 
induction in $\LAP{v}{v'}$,
using Lemmas~\ref{lem:selectB}, \ref{lem:assign-cont} and
\ref{lem:rassign-cont}, and the fact that the composition of two continuous
functions is continuous.

Thus $\HHH$ is a mapping from
$\PD \rightarrow \Dist \contarrow \Dist$ to itself.
The claim now follows since
all functions from the discrete cpo $\PD$ are continuous,
as a chain in $\PD$ can contain only one element.
\end{proof}
\begin{lemma}
\label{lem:HHcont}
The functional $\HHH$ is continuous
on $\PD \contarrow (\Dist \contarrow \Dist)$.
\end{lemma}
\begin{prooff}
Consider a chain $\chain{g_k}{k}$,
so as to prove that
$\HH{X}(\limit{k}{g_k}) = \limit{k}{\HH{X}(g_k)}$. For all $(v,v') \in \PD$
and all $D$ in $\Dist$, we must thus prove
\[
\HH{X}(\limit{k}{g_k})(v,v')(D) = \limit{k}{\HH{X}(g_k)(v,v')(D)}
\]
and shall do so by induction in $\LAP{v}{v'}$, with a case analysis 
in Definition~\ref{def:HHX}. We consider some typical cases:
\begin{itemize}
\item
If $\labv{v}$ is of the form $\assignv{x}{E}$
then both sides evaluate to
$\assignf{x}{E}(D)$.
\item
If $v' \neq v''$ where $v'' = \fppd{v}$
then we have the calculation 
\begin{eqnarray*}
\HHH(\limit{k}{g_k})(v,v')(D) 
& = &
\HHH(\limit{k}{g_k})(v'',v')(\HHH(\limit{k}{g_k})(v,v'')(D))
 \\ & = &
(\limit{k}{\HHH(g_k)(v'',v')})(\limit{k}{\HHH(g_k)(v,v'')(D)})
 \\ & = &
\limit{k}{\left(\HHH(g_k)(v'',v')(\HHH(g_k)(v,v'')(D))\right)}
 \\ & = &
\limit{k}{\HHH(g_k)(v,v')(D)}
\end{eqnarray*}
where the second equality follows from the induction hypothesis,
and the third equality from continuity of $\HHH(g_k)(v'',v')$
(Lemma~\ref{lem:hcontH}).
\item
If $v$ is a branching node with condition $B$,
$\true$-successor $v_1$, and $\false$-successor $v_2$,
where $\LAP{v_1}{v'} \geq \LAP{v}{v'}$ and
$\LAP{v_2}{v'} < \LAP{v}{v'}$ (other cases are similar),
with $D_1 = \selectf{B}(D)$ and $D_2 = \selectf{\neg B}(D)$
we have the calculation (where the second equality follows
from the induction hypothesis):
\begin{eqnarray*}
\HHH(\limit{k}{g_k})(v,v')(D) 
 & = &
\limit{k}{g_k}(v_1,v')(D_1) + \HHH(\limit{k}{g_k})(v_2,v')(D_2)
 \\ & = &
\limit{k}{(g_k(v_1,v')(D_1) + \HHH(g_k)(v_2,v')(D_2))}
 \\ & = &
\limit{k}{\HHH(g_k)(v,v')(D)}.
\end{eqnarray*}
\end{itemize}
\end{prooff}
\begin{proposition}
The functional $\HHH$ has a least fixed point
(belonging to $\PD \contarrow (\Dist \contarrow \Dist)$),
called $\fixed{\HHH}$, and given as
$\limit{k}{\HHH^k(0)}$, that is
the limit of the chain $\chain{\HHH^k(0)}{k}$
(where $\HHH^k(0)$ denotes $k$ applications of
$\HHH$ to the modification function that maps all distributions to 0). 
\end{proposition}
\begin{proof}
This follows from Lemma~\ref{lem:cont-cpo-fix},
using Lemmas~\ref{lem:HHcont} and \ref{lem:space-cpo}.
\end{proof}
We can now define the meaning of a pCFG:
\begin{definition}[Meaning of Probabilistic Control Flow Graph]
\label{def:omega}
Given a pCFG, 
we define its meaning $\omega$ as $\omega = \fixed{\HHH}$.
That is, $\omega = \limit{k}{\omega_k}$ 
where $\omega_k = \HHH^k(0)$ (thus $\omega_0 = 0$).
\end{definition}
Thus for all $k > 0$ we have $\omega_k = \HHH(\omega_{k-1})$.
Intuitively speaking, $\omega_k$ is the meaning of the pCFG assuming
that control is allowed to loop, that is move ``backwards'',
at most $k-1$ times.
\begin{lemma}
\label{lem:k-fixed-mult-nonincr-determ} 
For each $(v,v') \in \PD$, and each $k \geq 0$,
it holds that $\hpd{\omega_k}{v}{v'}$ is
additive, multiplicative and non-increasing;
it is even deterministic if the pCFG is deterministic.
\end{lemma}
\begin{proof}
We do induction in $k$, with the case $k = 0$ trivial
as the function $0$ is obviously additive, multiplicative and non-increasing,
and deterministic.

For the inductive step, we have to prove that the functional $\HHH$
preserves the property of being additive, multiplicative, non-increasing,
and (if the pCFG is deterministic) being deterministic.
But that is an easy induction in $\LAP{v}{v'}$,
using Lemmas~\ref{lem:selectB}, \ref{lem:assign-misc} and
\ref{lem:rassign-misc}.
\end{proof}

\begin{lemma}
\label{lem:fixed-mult-nonincr-determ}
For each $(v,v') \in \PD$,
it holds that $\hpd{\omega}{v}{v'}$ is
additive, multiplicative and non-increasing;
it is even deterministic if the pCFG is deterministic.
\end{lemma}
\begin{proof}
The claim follows from Lemma~\ref{lem:k-fixed-mult-nonincr-determ}
and the fact that when $f = \limit{k}{f_k}$ then
\begin{itemize}
\item
if each $f_k$ is additive then $f$ is additive, since
$f(D_1 + D_2) = \limit{k}{f_k(D_1 + D_2)} = \limit{k}{(f_k(D_1) + f_k(D_2))} 
= \limit{k}{f_k(D_1)} + \limit{k}{f_k(D_2)} = f(D_1) + f(D_2)$;
\item
if each $f_k$ is multiplicative then $f$ is multiplicative, since
$f(cD) = \limit{k}{f_k(cD)} = \limit{k}{c f_k(D)} 
= c\;\limit{k}{f_k(D)} = cf(D)$;
\item
if each $f_k$ is non-increasing then $f$ is non-increasing, since
(by Lemma~\ref{lem:sumlim-limsum})
$\sumd{f(D)} = \sumd{\limit{k}{f_k(D)}} 
= \limit{k}{\sumd{f_k(D)}}
\leq \limit{k}{\sumd{D}} = \sumd{D}$;
\item
if each $f_k$ is deterministic then $f$ is deterministic: to see this,
let concentrated $D$ be given;
we must prove that $f(D)$ is concentrated.
If $f(D) = 0$, the claim is obvious.
Otherwise, there exists $m$ and $s_0 \in \fulls$ such 
that $f_m(D)(s_0) > 0$. As $\chain{f_k}{k}$ is a chain, we infer that
for all $n \geq m$ we have $f_n(D)(s_0) > 0$,
which as each $f_n$ is deterministic implies
that for all $s \in \fulls$ with $s \neq s_0$ we have
$f_n(D)(s) = 0$, and thus the desired $f(D)(s) = 0$.
\end{itemize}
\end{proof}

\subsubsection{Examples}
\label{subsec:sem-ex}

\begin{figure}
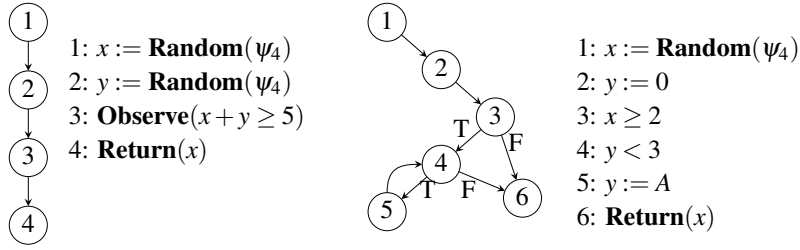

\begin{pgfpicture}{0mm}{5mm}{130mm}{35mm}
\begin{pgfmagnify}{0.9}{0.9}

\pgfnodecircle{v1}[stroke]{\pgfxy(2.5,3.5)}{8pt}

\pgfputat{\pgfxy(2.5,3.4)}{\pgfbox[center,base]{1}}

\pgfnodecircle{v2}[stroke]{\pgfxy(2.5,2.5)}{8pt}

\pgfputat{\pgfxy(2.5,2.4)}{\pgfbox[center,base]{2}}

\pgfnodecircle{v3}[stroke]{\pgfxy(2.5,1.5)}{8pt}

\pgfputat{\pgfxy(2.5,1.4)}{\pgfbox[center,base]{3}}

\pgfnodecircle{v4}[stroke]{\pgfxy(2.5,0.5)}{8pt}

\pgfputat{\pgfxy(2.5,0.4)}{\pgfbox[center,base]{4}}

\pgfsetarrowsend{stealth}
\pgfsetendarrow{\pgfarrowtriangle{4pt}}

\pgfnodeconnline{v1}{v2}

\pgfnodeconnline{v2}{v3}

\pgfnodeconnline{v3}{v4}

\pgfnodecircle{w1}[stroke]{\pgfxy(7.8,3.5)}{8pt}

\pgfputat{\pgfxy(7.8,3.4)}{\pgfbox[center,base]{1}}

\pgfnodecircle{w2}[stroke]{\pgfxy(8.6,2.8)}{8pt}

\pgfputat{\pgfxy(8.6,2.7)}{\pgfbox[center,base]{2}}

\pgfnodecircle{w3}[stroke]{\pgfxy(9.4,2.1)}{8pt}

\pgfputat{\pgfxy(9.4,2.0)}{\pgfbox[center,base]{3}}

\pgfnodecircle{w4}[stroke]{\pgfxy(8.6,1.4)}{8pt}

\pgfputat{\pgfxy(8.6,1.3)}{\pgfbox[center,base]{4}}

\pgfnodecircle{w5}[stroke]{\pgfxy(7.8,0.7)}{8pt}

\pgfputat{\pgfxy(7.8,0.6)}{\pgfbox[center,base]{5}}

\pgfnodecircle{w6}[stroke]{\pgfxy(9.8,0.9)}{8pt}

\pgfputat{\pgfxy(9.8,0.8)}{\pgfbox[center,base]{6}}

\pgfsetarrowsend{stealth}
\pgfsetendarrow{\pgfarrowtriangle{4pt}}

\pgfnodeconnline{w1}{w2}

\pgfnodeconnline{w2}{w3}

\pgfnodeconnline{w3}{w4}

\pgfputat{\pgfxy(8.9,1.8)}{\pgfbox[center,base]{T}}

\pgfputat{\pgfxy(9.7,1.6)}{\pgfbox[center,base]{F}}

\pgfnodeconnline{w3}{w6}

\pgfnodeconnline{w4}{w5}

\pgfputat{\pgfxy(8.4,0.9)}{\pgfbox[center,base]{T}}

\pgfputat{\pgfxy(9.0,0.9)}{\pgfbox[center,base]{F}}

\pgfnodeconnline{w4}{w6}

\pgfnodeconncurve{w5}{w4}{90}{180}{10pt}{10pt}

\pgfputat{\pgfxy(3.1,3.0)}{\pgfbox[left,base]{1: $\rassignv{x}{\psi_4}$}}
\pgfputat{\pgfxy(3.1,2.5)}{\pgfbox[left,base]{2: $\rassignv{y}{\psi_4}$}}
\pgfputat{\pgfxy(3.1,2.0)}{\pgfbox[left,base]{3: $\observev{x+y \geq 5}$}}
\pgfputat{\pgfxy(3.1,1.5)}{\pgfbox[left,base]{4: $\retv{x}$}}

\pgfputat{\pgfxy(10.6,3.0)}{\pgfbox[left,base]{1: $\rassignv{x}{\psi_4}$}}
\pgfputat{\pgfxy(10.6,2.5)}{\pgfbox[left,base]{2: $\assignv{y}{0}$}}
\pgfputat{\pgfxy(10.6,2.0)}{\pgfbox[left,base]{3: $x \geq 2$}}
\pgfputat{\pgfxy(10.6,1.5)}{\pgfbox[left,base]{4: $y < 3$}}
\pgfputat{\pgfxy(10.6,1.0)}{\pgfbox[left,base]{5: $\assignv{y}{A}$}}
\pgfputat{\pgfxy(10.6,0.5)}{\pgfbox[left,base]{6: $\retv{x}$}}

\end{pgfmagnify}
\end{pgfpicture}

\caption{\label{fig:ex12r} The pCFGs $G_1$ and $G_2$ (copied
from Figure~\ref{fig:ex12}).}
\end{figure}

We now illustrate our semantics on the pCFGs from
Figure~\ref{fig:ex12} 
(which for the reader's convenience we
have copied into Figure~\ref{fig:ex12r}); for both pCFGs,
we assume that $\uvar = \{x,y\}$.

\paragraph{Analyzing $G_1$}
Given an arbitrary
distribution $D$ with $\sumd{D} = 1$, let $D_1 = \hpd{\omega}{1}{2}(D)$,
and let $D_2 = \hpd{\omega}{1}{3}(D) = \hpd{\omega}{2}{3}(D_1)$.
We now have the calculation
\begin{eqnarray*}
D_2(s_2) & = & 
\psi_4(s_2(y)) \left( \sum_{s_1 \in \fulls\ \mid\ 
s_2 = \upd{s_1}{y}{s_2(y)}}  D_1(s_1) \right)
\\
& = &
\psi_4(s_2(y)) \left( \sum_{s_1 \in \fulls\ \mid\ 
s_2 = \upd{s_1}{y}{s_2(y)}}  \psi_4(s_1(x))
\left( \sum_{s \in \fulls\ \mid\ s_1 = \upd{s}{x}{s_1(x)}} D(s) \right) \right)
\\
& = &
\psi_4(s_2(y)) \left( \sum_{s_1 \in \fulls\ \mid\ 
s_2 = \upd{s_1}{y}{s_2(y)}}  \psi_4(s_2(x))
\left( \sum_{s \in \fulls\ \mid\ s_1 = \upd{s}{x}{s_2(x)}} D(s) \right) \right)
\\
& = &
\psi_4(s_2(y)) \psi_4(s_2(x)) \left( \sum_{s,s_1 \in \fulls\ \mid\ 
s_2 = \upd{\upd{s}{x}{s_2(x)}}{y}{s_2(y)},\ 
s_1 = \upd{s}{x}{s_2(x)}} D(s) \right)
\\
& = &
\psi_4(s_2(y)) \psi_4(s_2(x)) \left( \sum_{s \in \fulls\ \mid\ 
s_2 = \upd{\upd{s}{x}{s_2(x)}}{y}{s_2(y)}} D(s) \right)
\\
& = &
\psi_4(s_2(y)) \psi_4(s_2(x)) \left( \sum_{s \in \fulls} D(s) \right)
=
\psi_4(s_2(y)) \psi_4(s_2(x))
\end{eqnarray*}
from which we conclude that $D_2(\stotwo{x}{i}{y}{j}) =
\frac{1}{16}$ for $i,j \in \{0..3\}$, and $D_2(s) = 0$ otherwise.

The final distribution $D' = \hpd{\omega}{1}{4}(D)$
is given by $D' = \hpd{\omega}{3}{4}(D_2)$. Hence
(thus $\displaystyle \sumd{D'} = \frac{3}{16}$)
\[
D'(\stotwo{x}{2}{y}{3}) = D'(\stotwo{x}{3}{y}{2})
= D'(\stotwo{x}{3}{y}{3}) = \frac{1}{16} \mbox{ (0 otherwise).}
\]
\paragraph{Analyzing $G_2$}
We restrict our attention to
$\hpd{\omega}{4}{6}$, in particular when applied to a 
generic concentrated
distribution $\Dijr{i}{j}{r}$ 
defined by stipulating
\begin{eqnarray*}
\Dijr{i}{j}{r}(s) & = & r \mbox{ if } s = \stotwo{x}{i}{y}{j} \\
\Dijr{i}{j}{r}(s) & = & 0 \mbox{ otherwise}
\end{eqnarray*}
(Observe that when control first reaches node 4, the distribution is
$\Dijr{2}{0}{0.25} + \Dijr{3}{0}{0.25}$
since $y$ is zero, and for $x$, only the values 2 and 3 lead to
node $4$ while the values 0 and 1 do not.)
Before looking at the various possibilities for the assignment
at node 5, let us do some general reasoning.

Clause~\ref{evalk-transitive} 
in Definition~\ref{def:HHX} (and the definition of $\omega_0$) gives us
\begin{equation}
\label{eq:Dijr-trans}
\forall k \geq 0, \forall D \in Dist:\
\hpd{\omega_k}{5}{6}(D) = \hpd{\omega_k}{4}{6}(\hpd{\omega_k}{5}{4}(D)).
\end{equation}
Observe that
if $j \geq 3$ then $\selectf{y < 3}(\Dijr{i}{j}{r}) = 0$ and 
$\selectf{\neg(y < 3)}(\Dijr{i}{j}{r}) = \Dijr{i}{j}{r}$,
whereas if $j < 3$ then
$\selectf{\neg(y < 3)}(\Dijr{i}{j}{r}) = 0$
and
$\selectf{y < 3}(\Dijr{i}{j}{r}) = \Dijr{i}{j}{r}$.

Since $\LAP{5}{6} > \LAP{4}{6}$, 
clause~\ref{def:evalk:cond} in Definition~\ref{def:HHX}
(substituting $\omega_{k-1}$ for $h_0$)
gives us for $j < 3$ and $k \geq 1$:
$\hpd{\omega_k}{4}{6}(\Dijr{i}{j}{r}) =
\hpd{\omega_{k-1}}{5}{6}(\Dijr{i}{j}{r}) +
\hpd{\omega_k}{6}{6}(0) = 
\hpd{\omega_{k-1}}{5}{6}(\Dijr{i}{j}{r})$
and thus
\begin{equation}
\label{eq:Dijr-true}
\forall j < 3, \forall k \geq 1:\
\hpd{\omega_k}{4}{6}(\Dijr{i}{j}{r}) = 
\hpd{\omega_{k-1}}{5}{6}(\Dijr{i}{j}{r}).
\end{equation}
Similarly, for $j \geq 3$ and $k \geq 1$ we have
$\hpd{\omega_k}{4}{6}(\Dijr{i}{j}{r}) =
\hpd{\omega_{k-1}}{5}{6}(0) +
\hpd{\omega_k}{6}{6}(\Dijr{i}{j}{r}) = 0 + \Dijr{i}{j}{r}$; thus
\begin{equation}
\label{eq:Dijr-false} 
\forall j \geq 3, \forall k \geq 1:\
\hpd{\omega_k}{4}{6}(\Dijr{i}{j}{r}) = \Dijr{i}{j}{r}.
\end{equation}
We now look at the various cases for the assignment at node 5
(always assuming $i \in \{0..3\}$).
\begin{description}
\item[$\assignv{y}{1}$]
For all $k \geq 1$, and all $j < 3$, we have
$\hpd{\omega_k}{5}{4}(\Dijr{i}{j}{r}) = \Dijr{i}{1}{r}$
and by (\ref{eq:Dijr-trans}) thus
$\hpd{\omega_k}{5}{6}(\Dijr{i}{j}{r}) = \hpd{\omega_k}{4}{6}(\Dijr{i}{1}{r})$
(which also holds for $k \geq 0$).
Thus from (\ref{eq:Dijr-true}) we get that
$\hpd{\omega_k}{4}{6}(\Dijr{i}{j}{r}) = 
\hpd{\omega_{k-1}}{4}{6}(\Dijr{i}{1}{r})$ for all $k \geq 1$ and $j < 3$.
As $\omega_0 = 0$, we see by induction that
$\hpd{\omega_k}{4}{6}(\Dijr{i}{j}{r}) = 0$ 
for all $k \geq 0$ and $j < 3$, and for all $j < 3$ we thus have
\[
\hpd{\omega}{4}{6}(\Dijr{i}{j}{r}) = 0
\]
which confirms that from node 4 
the probability of termination is zero (actually termination is impossible)
and that certainly $\hpd{\omega}{4}{6}$ is not sum-preserving.
\item[$\assignv{y}{y+1}$]
For all $k \geq 1$, and all $j < 3$, we have
$\hpd{\omega_k}{5}{4}(\Dijr{i}{j}{r}) = \Dijr{i}{j+1}{r}$
and by (\ref{eq:Dijr-trans}) thus
$\hpd{\omega_k}{5}{6}(\Dijr{i}{j}{r}) = \hpd{\omega_k}{4}{6}(\Dijr{i}{j+1}{r})$
(which also holds for $k \geq 0$).
Thus from (\ref{eq:Dijr-true}) we get that
$\hpd{\omega_k}{4}{6}(\Dijr{i}{j}{r}) = 
\hpd{\omega_{k-1}}{4}{6}(\Dijr{i}{j+1}{r})$ for all $k \geq 1$ and $j < 3$.
We infer by~(\ref{eq:Dijr-false}) that for
all $j < 3$, and all $k > 3-j$, 
\[
\hpd{\omega_k}{4}{6}(\Dijr{i}{j}{r}) =
\hpd{\omega_{k-(3-j)}}{4}{6}(\Dijr{i}{3}{r}) = \Dijr{i}{3}{r}
\]
and thus we infer that for all $j < 3$ we have
\[
\hpd{\omega}{4}{6}(\Dijr{i}{j}{r}) = \Dijr{i}{3}{r}
\]
which confirms that when $y < 3$ the loop terminates with $y = 3$.
\item[$\rassignv{y}{\psi_4}$]
For all $k \geq 1$, and all $j < 3$, 
for $s'$ with $s'(x) = i$ and $s'(y) \in \{0..3\}$ we have
\[
\hpd{\omega_k}{5}{4}(\Dijr{i}{j}{r})(s')
=
\psi_4(s'(y)) \sum_{s \in \fulls\ \mid\ s' = \upd{s}{y}{s'(y)}} \Dijr{i}{j}{r}(s)
=
0.25 \cdot r =
0.25 \cdot \sum_{q = 0,1,2,3} \Dijr{i}{q}{r}(s')
\]
from which we infer that
\[
\hpd{\omega_k}{5}{4}(\Dijr{i}{j}{r}) = 
0.25 \cdot \sum_{q = 0,1,2,3} \Dijr{i}{q}{r}
\]
and by (\ref{eq:Dijr-trans}), together with the fact
(Lemma~\ref{lem:k-fixed-mult-nonincr-determ}) that
$\hpd{\omega_k}{4}{6}$ is additive and multiplicative, thus
\[
\hpd{\omega_k}{5}{6}(\Dijr{i}{j}{r}) = 
0.25 \cdot \sum_{q = 0,1,2,3} (\hpd{\omega_k}{4}{6}(\Dijr{i}{q}{r})
\]
(which also holds for $k \geq 0$)
so from (\ref{eq:Dijr-true}) we get that
\begin{equation}
\label{eq:Dijr-recur}
\forall k \geq 1, j < 3:
\hpd{\omega_k}{4}{6}(\Dijr{i}{j}{r}) = 
0.25 \cdot \left(\sum_{q = 0,1,2,3}\hpd{\omega_{k-1}}{4}{6}(\Dijr{i}{q}{r})\right).
\end{equation}
One can easily prove by induction in $k$ that if
$j_1 < 3$ and $j_2 < 3$ then
$\hpd{\omega_k}{4}{6}(\Dijr{i}{j_1}{r}) =
\hpd{\omega_k}{4}{6}(\Dijr{i}{j_2}{r})$
so if we define
$D_k = \hpd{\omega_k}{4}{6}(\Dijr{i}{0}{r})$
we have
$\hpd{\omega_k}{4}{6}(\Dijr{i}{j}{r}) = D_k$ for all $j < 3$.
We now establish
\begin{equation}
\label{eq:Dijr-limit}
\limit{k}{D_k} = \Dijr{i}{3}{r}
\end{equation}
which will demonstrate
that a loop from node 4 will terminate,
with $y = 3$, with probability 1.

To show (\ref{eq:Dijr-limit}), observe that 
(\ref{eq:Dijr-recur}) together with (\ref{eq:Dijr-false})
makes it easy to prove by induction that
$D_k(s) = 0 = \Dijr{i}{3}{r}(s)$ 
for all $k \geq 0$ when
$s  \neq \stotwo{x}{i}{y}{3}$, 
and also gives the recurrences
\begin{eqnarray*}
D_0(s_3) & = & 0 \\
D_1(s_3) & = & 0 \\
D_k(s_3) & = & 0.75 \cdot D_{k-1}(s_3) + 0.25 \cdot r \mbox{ for }
k \geq 2
\end{eqnarray*}
when $s_3 = \stotwo{x}{i}{y}{3}$.
We must prove that $\limit{k}{D_k(s_3)} = r$
(as $\Dijr{i}{3}{r}(s_3) = r$)
but this follows, with $a = 0.75$ and $b = 0.25$, from a general result: \begin{lemma}
If $\chain{x_i}{i}$ is a sequence of non-negative reals, satisfying 
$x_0 = x_1 = 0$ and $x_k = a\,x_{k-1} +b\,r$ for $k > 1$ 
where $a,b,r$ are non-negative reals
with $b > 0$ and $a + b = 1$, then $\limit{i}{x_i} = r$.
\end{lemma}
\begin{proof}
Observe that: \emph{(i)} $\chain{x_i}{i}$ is a chain
(as can be seen by induction since $x_0 = x_1 \leq x_2$ and
if $x_k \leq x_{k+1}$ then $x_{k+1} \leq x_{k+2}$);
\emph{(ii)} $x_i \leq r$ for all $i$ since if $x_k > r$ for some $k$
then  $x_{k+1} = ax_k + br = (1-b)x_k + br = x_k + b(r-x_k) < x_k$
which contradicts $\chain{x_i}{i}$ being a chain;
\emph{(iii)} thus $\limit{i}{x_i} < \infty$ and
since $\limit{i}{x_i} = a \cdot \limit{i}{x_i} + br$
we get $b \cdot \limit{i}{x_i} = (1-a)\limit{i}{x_i} = br$
from which we infer the desired $\limit{i}{x_i} = r$.
\end{proof}
\end{description}

\section{The Structured Probabilistic Language}
We now present the language for probabilistic programs
used in, \eg, \cite{Gor+etal:ICSE-2014,Hur+etal:PLDI-2014}.
This language is structured (unlike the pCFG-based language presented
in Section~\ref{sec:cfg-lang}),
that is a program is built compositionally from constructs
that combine basic constructs.

\subsection{Syntax}
\label{subsec:struct-syntax}
\begin{figure}
\begin{eqnarray*}
S & ::= & \skipc \\
 & \mid & \assignv{x}{E} \\
 & \mid & \rassignv{x}{\psi} \\
 & \mid & \observev{B} \\
 & \mid & \seqc{S_1}{S_2} \\
 & \mid & \ifc{B}{S_1}{S_2} \\
 & \mid & \whilec{B}{S} 
\\[2mm]
P & ::= & \progc{S}{E}
\end{eqnarray*}
\caption{\label{fig:struct-lang} The grammar defining a structured probabilistic statement, and program.}
\end{figure}

A program is a \emph{statement} followed by the return of an expression,
where a statement $S$ is defined by the BNF
in Figure~\ref{fig:struct-lang}
(the syntactic details differ slightly from
what is presented in \cite{Gor+etal:ICSE-2014,Hur+etal:PLDI-2014}).
That is, a structured statement $S$ 
is either $\skipc$, an assignment
$\assignv{x}{E}$, a random assignment $\rassignv{x}{\psi}$,
a conditioning statement $\observev{B}$,
a sequential composition $\seqc{S_1}{S_2}$,
a conditional $\ifc{B}{S_1}{S_2}$,
or a while loop $\whilec{B}{S}$.

\subsection{Translation to pCFG Language}
We now present a translation $\trl$
from probabilistic structured programs
(Section~\ref{subsec:struct-syntax})
to pCFGs
(Section~\ref{sec:CFG}).
Recall that a pCFG has a special node $\startv$ 
from which there is a path to all other nodes,
and a unique $\finalv$ node with no outgoing edges
(if it has a label
it will be of the form $\retv{x}$)
to which there is a path from all other nodes.
\begin{definition}[Translation from Structured Statements to pCFGs]
\label{def:translateS}
For a structured statement $S$, we define 
a pCFG $\trl(S)$ whose $\finalv$ node has no label,
by structural induction in $S$:
\begin{itemize}
\item 
if $S$ is of the form $\skipc$, or
$\assignv{x}{E}$, or
$\rassignv{x}{\psi}$, or
$\observev{B}$, then $\trl(S)$ is a pCFG with 2 nodes:
$\startv$ with label as $S$, and $\finalv$ with no label;
there is an edge from the $\startv$ node to the $\finalv$ node.
\item
if $S$ is of the form $\seqc{S_1}{S_2}$, we inductively construct a pCFG 
$\trl(S_1)$ with $\startv$ node $v_1$ and unlabeled 
$\finalv$ node $v'_1$, 
and a pCFG 
$\trl(S_2)$ with $\startv$ node $v_2$ and unlabeled $\finalv$ node $v'_2$.
The pCFG $\trl(\seqc{S_1}{S_2})$ is then constructed 
by taking the union of $\trl(S_1)$ and $\trl(S_2)$
(which must have disjoint node sets),
and augmenting the result as follows:
\begin{enumerate}
\item
let the $\startv$ node be $v_1$;
\item
let the $\finalv$ node be $v'_2$;
\item
give $v'_1$ the label $\skipv$,
and add an edge from $v'_1$ to $v_2$.
\end{enumerate}
\item
if $S$ is of the form $\ifc{B}{S_1}{S_2}$, we inductively construct a pCFG 
$\trl(S_1)$ with $\startv$ node $v_1$ and unlabeled 
$\finalv$ node $v'_1$, 
and a pCFG 
$\trl(S_2)$ with $\startv$ node $v_2$ and unlabeled $\finalv$ node $v'_2$.
The pCFG $\trl(\ifc{B}{S_1}{S_2})$ is then constructed 
by taking the union of $\trl(S_1)$ and $\trl(S_2)$
(which must have disjoint node sets),
and augmenting the result as follows:
\begin{enumerate}
\item
let the $\startv$ node be a \emph{fresh} branching node $v$ 
with condition $B$,
$\true$-successor $v_1$, and $\false$-successor $v_2$;
\item
let the $\finalv$ node be a \emph{fresh} unlabeled node $v'$;
\item
give $v'_1$ and $v'_2$ the label $\skipv$,
and add edges from $v'_1$ to $v'$ and from $v'_2$ to $v'$.
\end{enumerate}
\item
if $S$ is of the form $\whilec{B}{S_1}$, we inductively construct a pCFG 
$\trl(S_1)$ with $\startv$ node $v_1$ and unlabeled 
$\finalv$ node $v'_1$.
The pCFG $\trl(\whilec{B}{S_1})$ is then constructed
by augmenting $\trl(S_1)$ as follows:
\begin{enumerate}
\item
let the $\finalv$ node be a \emph{fresh} unlabeled node $v'$;
\item
let the $\startv$ node be a \emph{fresh} branching node $v$
with condition $B$,
$\true$-successor $v_1$, and $\false$-successor $v'$;
\item
give $v'_1$ the label $\skipv$, and add an edge
from $v'_1$ to $v$.
\end{enumerate}
\end{itemize}
\end{definition}
It is easy to verify by induction that the pCFGs constructed
by Definition~\ref{def:translateS} are indeed well-formed,
in particular,
that all nodes are reachable
from the $\startv$ node, and can reach the $\finalv$ node.
\begin{definition}[Translation from Structured Programs to pCFGs]
\label{def:translateP}
For a structured program $P \equiv \progc{S}{E}$,
we define a pCFG $\trl(P)$ as follows:
first construct the pCFG $\trl(S)$; 
then label its $\finalv$ node (unlabeled so far) with
$\retv{E}$.
\end{definition}
\begin{example}
\label{ex:p1}
Consider the pCFG $G_1$ depicted in Figure~\ref{fig:ex12}(left).
We have $G_1 = \trl(P_1)$ where
\[
P_1 \eqdef \progc{\seqc{\seqc{\rassignv{x}{\psi_4}}{\rassignv{y}{\psi_4}}}{\observev{x+y \geq 5}}}{x}
\]
\end{example}
\begin{example}
\label{ex:p2}
Let the structured program $P_2$ be given by
\[
P_2 \eqdef
\progc{\seqc{\rassignv{x}{\psi_4}}{\seqc{\assignv{y}{0}}{\ifc{x \geq 2}{\whilec{y < 3}{\assignv{y}{A}}}{\skipc}}}}{x}
\]
It is easy to see that $\trl(P_2)$ can be simplified
(by compression of edges from nodes labeled $\skipv$)
into a pCFG that is isomorphic to $G_2$ 
depicted in Figure~\ref{fig:ex12}(right).
\end{example}
For a given pCFG $G$, there may not exist a structured program $P$
such that $G$ is isomorphic to a simplification of $\trl(P)$.
A necessary condition is that 
$G$ is ``reducible''.

\subsection{Semantics}
\label{subsec:struct-sem}
We now present,
following~\cite{Gor+etal:ICSE-2014,Hur+etal:PLDI-2014},
the semantics of the structured language.
The semantics manipulates ``expectation functions'' where an
expectation function $F$ is a function from stores to non-negative
reals; we can think of $F(s)$ as the expected return value for store
$s$. The semantics of a statement $S$, written $\semc{S}$,
is a transformation of expectation functions; with $\semc{S} F' = F$,
one should think
of $F'$ as taking a store \emph{after} $S$ and giving 
its expected return value, and $F$ as taking a store 
\emph{before} $S$ and giving its expected return value.

\begin{figure}

\begin{eqnarray*}
F = \semc{\skipc} F'  & \mbox{iff} &  F = F'
\\[2mm]
F = \semc{\assignv{x}{E}}F' & \mbox{iff} & F(s) = F'(\upd{s}{x}{\seme{E}s})
\mbox{ for all } s
\\[2mm]
F = \semc{\rassignv{x}{\psi}}F' & \mbox{iff} & 
F(s) = \sum_{z \in \Zz}\psi(z)F'(\upd{s}{x}{z}) \mbox{ for all } s
\\[2mm]
F = \semc{\observev{B}}F' & \mbox{iff} & F(s) =
F'(s) \mbox{ for all $s$ with } \seme{B}s \\[1mm]
 & \mbox{and} & F(s) = 0 \mbox{ for all other $s$}
\\[2mm]
F = \semc{\seqc{S_1}{S_2}}F' & \mbox{iff} &  
F = \semc{S_1}(\semc{S_2}F')
\\[2mm]
F = \semc{\ifc{B}{S_1}{S_2}}F' & \mbox{iff} & 
F(s) = \semc{S_1}(F')(s) \mbox{ for all $s$ with } \seme{B}s \\[1mm]
& \mbox{and} & F(s) = \semc{S_2}(F')(s) \mbox{ for all other $s$}
\\[2mm]
F = \semc{\whilec{B}{S}}F' & \mbox{iff} & F(s) = \limit{k}{F_k(s)}
\mbox{ for all } s \\[1mm]
& \mbox{where} & F_0(s) = 0 \\[1mm]
& \mbox{and} & F_{k+1}(s) = \semc{S}(F_k)(s) \mbox{ if } \seme{B}s \\[1mm]
& \mbox{and} & F_{k+1}(s) = F'(s) \mbox{ otherwise}
\end{eqnarray*}
\caption{\label{fig:struct-sem} The semantics
of a structured probabilistic statement.}
\end{figure}

In Figure~\ref{fig:struct-sem}, we define $\semc{S}$ by a definition
inductive in $S$. Let us explain a few cases:
\begin{itemize}
\item
the expected return value for a store before
an assignment $\assignv{x}{E}$ 
equals the expected return value for the updated store;
\item
the expected return value for a store before
a random assignment $\rassignv{x}{\psi}$
can be found by taking the weighted average of the expected return
values for the possible updated stores;
\item
the expected return value for a store before
a conditioning statement $\observev{B}$ is 0 if $B$ is not true;
\item
the semantics of 
a while loop $\whilec{B}{S}$ can be found as the limit
of the semantics of the $k$th iteration, $\whileck{B}{S}{k}$,
which is defined inductively in $k$ as follows:
\begin{eqnarray*}
\whileck{B}{S}{0} & = & \observev{\false} 
\\[1mm]
\whileck{B}{S}{k+1} & = &
\ifc{B}{(\seqc{S}{\whileck{B}{S}{k}})}{\skipc}
\end{eqnarray*}
\end{itemize}

\begin{example}
\label{ex2:standard}
Consider the statement $S_1$ given by (cf.~Example~\ref{ex:p1})
\[
S_1 \equiv \rassignv{x}{\psi_4};\ \rassignv{y}{\psi_4};\ \observev{x+y \geq 5}.
\]
For all $F$ that map stores into non-negative reals,
and all stores $s$, we have
\begin{eqnarray*}
\semc{S_1} F\; s & = &
\semc{\rassignv{x}{\psi_4}} 
 (\semc{\rassignv{y}{\psi_4};\ \observev{x+y \geq 5}} F)\; s
\\ 
& = &
\sum_{q \in 0..3} \frac{1}{4} 
  (\semc{\rassignv{y}{\psi_4};\ \observev{x+y \geq 5}} F\; \upd{s}{x}{q})
\\ 
& = &
\sum_{q \in 0..3} \frac{1}{4} 
  (\sum_{q' \in 0..3} \frac{1}{4}
    (\semc{\observev{x+y \geq 5}} F\; \upd{\upd{s}{x}{q}}{y}{q'}))
\\ 
& = &
\frac{1}{16}
\left(F(\upd{\upd{s}{x}{2}}{y}{3}) +
F(\upd{\upd{s}{x}{3}}{y}{2}) +
F(\upd{\upd{s}{x}{3}}{y}{3})\,\right).
\end{eqnarray*}
\end{example}
For a program $P = \progc{S}{E}$, the expectation function at the
end will map $s$ into $\seme{E}s$,
and thus the expectation function at the beginning appears to
be given as $\semc{S}(\lambda s.\seme{E}s)$.
But this assumes that runs that fail conditioning
statements count as zero; such runs should rather not be taken
into account at all. This motivates the following 
definition~\cite{Gor+etal:ICSE-2014} of
the \emph{normalized} semantics of a structured program:
\begin{equation}
\label{eq:norm-sem}
\semc{\progc{S}{E}} = \frac{\semc{S}(\lambda s. \seme{E}{s})(\bot)}%
{\semc{S}(\lambda s. 1)(\bot)}
\end{equation}
where $\bot$ is an ``initial store''
(if we demand that all variables are defined before they are used
then the choice of initial store is irrelevant).

To illustrate this definition, let us
look at $P_1$ as defined in Example~\ref{ex:p1}.
Here $P_1 = \progc{S_1}{x}$ with $S_1$ defined as in
Example~\ref{ex2:standard}, and from that example we see 
(since $\seme{x}s = s(x)$) that
\begin{eqnarray*}
\semc{S_1} (\lambda s.\seme{x}s)\ \bot & = & \frac{1}{16} (2 + 3 + 3) = \frac{8}{16}
\\[1mm]
\semc{S_1} (\lambda s.1)\ \bot & = & \frac{1}{16} (1 + 1 + 1) = \frac{3}{16}
\end{eqnarray*}
and hence we see by (\ref{eq:norm-sem}) that
\[
\semc{P_1} = \frac{\semc{S_1}(\lambda s. \seme{x}{s})(\bot)}%
{\semc{S_1}(\lambda s. 1)(\bot)} = \mathbf{\frac{8}{3}}.
\]
This makes sense: if $P_1$ terminates then $x+y \geq 5$ which
holds in 3 cases; in two cases, $x = 3$ whereas in one case, $x = 2$,
for a weighted average of $2 \frac{2}{3}$.

\section{Adequacy Result for the Two Semantics}
\label{sec:consistent}

To motivate how the semantics in Section~\ref{subsec:struct-sem} 
relates to the semantics
in Section~\ref{sec:sem},
consider $S_1$ as defined in Example~\ref{ex2:standard}.
Then, cf.~Example~\ref{ex:p1}, $\trl(S_1)$ is the pCFG $G_1$ depicted
in Figure~\ref{fig:ex12}(left), except that node 4 is unlabeled.

In Example~\ref{ex2:standard} we saw that
if $F = \semc{S_1}F'$ for some $F'$ then for all stores $s$ we have
\[
F(s) = 
\frac{1}{16}
\left(F'(\upd{\upd{s}{x}{2}}{y}{3}) +
F'(\upd{\upd{s}{x}{3}}{y}{2}) +
F'(\upd{\upd{s}{x}{3}}{y}{3})\,\right).
\]
In Section~\ref{subsec:sem-ex},
we saw that if
$D' = \hpd{\omega}{1}{4}(D)$
for some $D$ with $\sumd{D} = 1$ then
\begin{eqnarray*}
D'(s) & = & \displaystyle \frac{1}{16} \mbox{ if }
s \in \{\stotwo{x}{2}{y}{3},\ \stotwo{x}{3}{y}{2},\ \stotwo{x}{3}{y}{3}\}
\\[1mm]
D'(s) & = & 0 \mbox{ otherwise}.
\end{eqnarray*}
We now observe that 
(with the first equality due to $F(s)$ not depending on $s$
as $\uvar = \{x,y\}$)
\begin{eqnarray*}
\sum_{s \in \fulls} F(s) D(s) & = &
 F(s) \sum_{s \in \fulls} D(s)
=
 F(s)
\\[1mm]
& = &
\frac{1}{16}
F'\stotwo{x}{2}{y}{3} +
\frac{1}{16}
F'\stotwo{x}{3}{y}{2} +
\frac{1}{16}
F'\stotwo{x}{3}{y}{3}
\\ & = &
\sum_{s \in \fulls} F'(s) D'(s).
\end{eqnarray*}
And this is indeed an instance of the general result
relating the two semantics:
\begin{thm}[adequacy]
\label{thm:semantics-consistent}
Let $S$ be a structured statement, and
let the pCFG $\trl(S)$ have
$\startv$ node $v$,
and (unlabeled) $\finalv$ node $v'$,
and meaning $\omega$ (cf.~Definition~\ref{def:omega}).

If $\semc{S}F' = F$ 
and $\hpd{\omega}{v}{v'}(D) = D'$ 
then
\[
\sum_{s \in \fulls} F(s) D(s) = \sum_{s \in \fulls} F'(s) D'(s).
\]
\end{thm}
This result shows that we do not lose any information by using the semantics 
in Section~\ref{sec:sem}, in that for any structured statement $S$,
and any expectation function $F'$,
we can retrieve $F = \semc{S} F'$ from $\omega$:
for given $s_0 \in \fulls$, 
define $D_0$ such that $D_0(s_0) = 1$ but $D_0(s) = 0$ otherwise;
with $D'_0 = \hpd{\omega}{v}{v'}(D_0)$ we then have
\begin{equation}
\label{eq:FDex}
F(s_0) = 
\sum_{s \in \fulls} F(s) D_0(s) =
\sum_{s \in \fulls} F'(s) D'_0(s).
\end{equation}
If the statement $S$ (and thus also the pCFG $\trl(S)$)
is deterministic, 
that is it contains
no random assignments or conditioning 
(cf.~Definition~\ref{def:determCFG}),
then by Lemma~\ref{lem:fixed-mult-nonincr-determ} 
we see that $\hpd{\omega}{v}{v'}$ is deterministic
(cf.~Section~\ref{subsec:sem-domain})
and thus there will be 
some $s'_0 \in \fulls$ such that $D'_0$ is concentrated on $s'_0$.
Then (\ref{eq:FDex}) gives
the equation 
\[
F(s_0) = F'(s'_0)D'_0(s'_0).
\]
Since it is easy to see 
(as when proving Lemma~\ref{lem:fixed-mult-nonincr-determ})
that for a pCFG without random assignments, if $D$ maps into integers
then also $f(D)$ will map into integers, there are two possibilities:
\begin{itemize}
\item
$D'_0(s'_0) = 0$, which will happen if the program loops
when run on input store $s_0$;
in that case, $F(s_0) = 0$ which reflects that then the expected
return value for $s_0$ is zero.
\item
$D'_0(s'_0) = 1$, which will happen if the program terminates on $s_0$;
in that case, $F(s_0) = F'(s'_0)$ which reflects that then the
expected return value for the 
initial store equals the actual return value in the final store.
\end{itemize}

\paragraph{Proof of Theorem~\ref{thm:semantics-consistent}}
We do structural induction in $S$, with a case analysis.
\begin{itemize}
\item
The case with $S = \skipc$ is trivial,
as then $F = F'$ and $D' = D$.
\item
For the case $S = \assignv{x}{E}$, 
the claim follows from the calculation
\begin{eqnarray*}
\sum_{s' \in \fulls} F'(s') D'(s')
& = & 
\sum_{s' \in \fulls} F'(s') 
\left( \sum_{s \in \fulls\ \mid\ s' = \upd{s}{x}{\seme{E}s}} D(s)\right)
= 
\sum_{s' \in \fulls, s \in \fulls\ \mid\ s' = \upd{s}{x}{\seme{E}s}} F'(s') D(s)
\\ & = &
\sum_{s \in \fulls} F'(\upd{s}{x}{\seme{E}s}) D(s)
=  \sum_{s \in \fulls} F(s) D(s).
\end{eqnarray*}
\item
For the case $S = \rassignv{x}{\psi}$,
the claim follows from the calculation
\begin{eqnarray*}
\sum_{s' \in \fulls} F'(s') D'(s')
& = &
\sum_{s' \in \fulls} F'(s') \left( \psi(s'(x))
\sum_{s \in \fulls\ \mid\ s' = \upd{s}{x}{s'(x)}} D(s) \right)
\\ & = &
\sum_{s,s' \in \fulls\ \mid\ s' = \upd{s}{x}{s'(x)}} 
F'(s') \psi(s'(x))D(s)
\\ & = &
\sum_{s \in \fulls}\sum_{z \in \Zz}\ \sum_{s' \in \fulls\ \mid\ s' = \upd{s}{x}{z}}
F'(s') \psi(s'(x)) D(s)
\\ & = &
\sum_{s \in \fulls}\sum_{z \in \Zz}
F'(\upd{s}{x}{z}) \psi(z) D(s)
 = 
\sum_{s \in \fulls} D(s) \sum_{z \in \Zz} F'(\upd{s}{x}{z}) \psi(z)
\\ & = &
\sum_{s \in \fulls} D(s) F(s).
\end{eqnarray*}
\item
For the case $S = \observev{B}$, 
the claim follows from the calculation
\begin{eqnarray*}
\sum_{s \in \fulls} F'(s) D'(s) & = & 
\sum_{s \in \fulls} F'(s) \selectf{B}(D)(s) 
= 
\sum_{s \in \fulls\ \mid\ \seme{B}s} F'(s) D(s) \\
& = &
\sum_{s \in \fulls} F(s)D(s).
\end{eqnarray*}
\end{itemize}
To handle the composite cases,
we need to ensure that the equation
$\hpd{\omega}{v}{v'} = \hpd{\omega}{v}{v''}; \hpd{\omega}{v''}{v'}$
(and similarly for $\omega_k$)
holds for all $(v,v''), (v'',v') \in \PD$.
(Here ``$;$'' denotes function composition: $(f ; g)(x) = g(f(x))$.)
Observe that from case~\ref{evalk-transitive} in 
Definition~\ref{def:HHX} we know only that the equation holds when
$v'' = \fppd{v}$.
\begin{lemma}
\label{lem:k-seq-comp}
Given a pCFG, let $\HHH$ 
be given as in Definition~\ref{def:HHX}.
For all $(v,v_1), (v_1,v_2) \in \PD$, and
for all $h_0 \in \PD \rightarrow \Dist \rightarrow \Dist$,
with $h = \HHH(h_0)$ we have
\[
\hh{v}{v_2} = \hh{v}{v_1} ; \hh{v_1}{v_2}.
\]
\end{lemma}
\begin{prooff}
The claim is by induction in $\LAP{v}{v_1}$.
If $v_1 = v$, the claim is obvious (as then
$\hh{v}{v_1}$ is the identity);
if $v_1 = \fppd{v}$, the equality follows from the definition
of $\HHH$.

So assume that with $v_0 = \fppd{v}$ we have $v_1 \neq v$ and $v_1 \neq v_0$.
By Lemma~\ref{lem:LAP-add}, $\LAP{v_0}{v_1} < \LAP{v}{v_1}$.
Inductively, we thus have
$\hh{v_0}{v_2} = \hh{v_0}{v_1} ; \hh{v_1}{v_2}$
which gives us the desired result:
\[
\hh{v}{v_2} =
\hh{v}{v_0} ; \hh{v_0}{v_2} =
\hh{v}{v_0} ; (\hh{v_0}{v_1}; \hh{v_1}{v_2}) =
\hh{v}{v_1} ; \hh{v_1}{v_2}.
\]
\end{prooff}
\begin{lemma}
\label{lem:fixed-seq-comp}
Given a pCFG, let $\omega$ 
be as in Definition~\ref{def:omega}.
For all $(v,v_1), (v_1,v_2) \in \PD$:
\[
\hpd{\omega}{v}{v_2} = \hpd{\omega}{v}{v_1} ; \hpd{\omega}{v_1}{v_2}.
\]
\end{lemma}
\begin{proof}
This follows from Lemma~\ref{lem:k-seq-comp} since
$\omega = \HHH(\omega)$.
\end{proof}
We now resume the proof of Theorem~\ref{thm:semantics-consistent},
giving the 3 composite cases. In each of them, we exploit that
if a pCFG $G_1$ with $\startv$ node $v_1$ and unlabeled
$\finalv$ node $v'_1$ is a subgraph of a pCFG $G$, where in $G$,
$v'_1$ has been labeled and an outgoing edge added, then
$\hpd{\omega}{v_1}{v'_1} = \hpd{\omega_1}{v_1}{v'_1}$
where $\omega$ is the meaning (cf.~Definition~\ref{def:omega})
of $G$, and $\omega_1$ is the meaning of $G_1$.
This follows since Definition~\ref{def:HHX} does not make
use of the label of $v'$ when defining $\hh{v}{v'}$.

\begin{itemize}
\item
For the case $S = \seqc{S_1}{S_2}$,
recall that the pCFG $G = \trl(S)$, with $\startv$ node $v$
and $\finalv$ node $v'$, is constructed 
by taking the union of 
$\trl(S_1)$, with $\startv$ node $v$ and $\finalv$ node $v'_1$,
and $\trl(S_2)$, with $\startv$ node $v_2$ and $\finalv$ node $v'$,
and then giving
$v'_1$ the label $\skipv$,
and adding an edge from $v'_1$ to $v_2$.
In $G$ it is thus the case that
$v'$ postdominates $v_2$, $v_2 = \fppd{v'_1}$, and 
$v'_1$ postdominates $v$.
Hence we infer, by Lemma~\ref{lem:fixed-seq-comp}, 
that there exists $D''$ such that
\[
D'' = \hpd{\omega}{v}{v'_1}(D) \mbox{ and } D' = \hpd{\omega}{v_2}{v'}(D'').
\]
From Figure~\ref{fig:struct-sem}, we see that that there exists 
$F''$ such that $F'' = \semc{S_2}F'$ and $F = \semc{S_1}F''$.
By applying the induction hypothesis to first $S_2$ and next $S_1$,
we get the desired
\[
\sum_{s \in \fulls} F'(s) D'(s) =
\sum_{s \in \fulls} F''(s) D''(s) =
\sum_{s \in \fulls} F(s) D(s).
\]

\item
For the case $S = \ifc{B}{S_1}{S_2}$, 
recall that the pCFG $G = \trl(S)$ is constructed 
by taking the union of 
$\trl(S_1)$, with $\startv$ node $v_1$ and $\finalv$ node $v'_1$,
and $\trl(S_2)$, with $\startv$ node $v_2$ and $\finalv$ node $v'_2$,
and then letting the $\startv$ node
$v$ be a branching node 
with condition $B$ and
$\true$-successor $v_1$ and $\false$-successor $v_2$,
and for each of $v'_1$ and $v'_2$: give it the label $\skipv$,
and add an edge to the $\finalv$ node $v'$.
In $G$ it is thus the case that $v' = \fppd{v}$.

Hence, with $D_1 = \selectf{B}(D)$ and $D_2 = \selectf{\neg B}(D)$,
and with $D'_1 = \hpd{\omega}{v_1}{v'_1}(D_1)$ and
$D'_2 = \hpd{\omega}{v_2}{v'_2}(D_2)$, 
by Lemma~\ref{lem:fixed-seq-comp}
we have the calculation
\[
D' =  \hpd{\omega}{v}{v'}(D) 
=
\hpd{\omega}{v_1}{v'}(D_1) + \hpd{\omega}{v_2}{v'}(D_2)
=
\hpd{\omega}{v_1}{v'_1}(D_1) +
\hpd{\omega}{v_2}{v'_2}(D_2)
=
D'_1 + D'_2.
\]
From Figure~\ref{fig:struct-sem} we see that 
$F(s) = \semc{S_1}(F')(s)$ if $\seme{B}s$ holds,
and $F(s) = \semc{S_2}(F')(s)$ otherwise.

By applying the induction hypothesis to $S_1$ and $S_2$,
the claim now follows from the calculation
\begin{eqnarray*}
\sum_{s \in \fulls} F'(s) D'(s) & = &
\sum_{s \in \fulls} F'(s) D'_1(s) +
\sum_{s \in \fulls} F'(s) D'_2(s)
\\ & = &
\sum_{s \in \fulls} (\semc{S_1}(F')(s) \cdot D_1(s)) +
\sum_{s \in \fulls} (\semc{S_2}(F')(s) \cdot D_2(s))
\\ & = &
\sum_{s \in \fulls\ \mid\ \seme{B}s} F(s) D(s) +
\sum_{s \in \fulls\ \mid\ \seme{\neg B}s} F(s) D(s)
= \sum_{s \in \fulls} F(s) D(s).
\end{eqnarray*}
\item
For the case $S = \whilec{B}{S_1}$, 
recall that the pCFG $G = \trl(S)$, with $\startv$ node $v$
and $\finalv$ node $v'$, is constructed 
by augmenting
$\trl(S_1)$, with $\startv$ node $v_1$ and $\finalv$ node $v'_1$,
as follows:
let $v$ be a branching node 
with condition $B$ and
$\true$-successor $v_1$ and $\false$-successor $v'$,
and add an edge from $v'_1$ (now labeled $\skipv$) to $v$.
In $G$ it is thus the case that $v' = \fppd{v}$,
and that $v$ postdominates $v'_1$ which postdominates $v_1$;
moreover, $\LAP{v}{v'} = 1 < \LAP{v_1}{v'}$.
For each $k \geq 0$ we thus have the calculation
\begin{eqnarray*}
  & & \hpd{\omega_{k+1}}{v}{v'}(D) \\
 & = & \hpd{\omega_{k+1}}{v'}{v'}(\selectf{\neg B}(D))
    + \hpd{\omega_k}{v_1}{v'}(\selectf{B}(D))
\\ & = &
   \selectf{\neg B}(D) + 
\hpd{\omega_k}{v}{v'}(\hpd{\omega_k}{v_1}{v'_1}(\selectf{B}(D))).
\end{eqnarray*}
where the first equality follows by clause~\ref{def:evalk:cond}
in Definition~\ref{def:HHX},
and where the second equality is obvious if $k = 0$ and
otherwise follows by 
Lemma~\ref{lem:k-seq-comp}
(since $\hpd{\omega_k}{v'_1}{v}$ is the identity).

For our proof, it is convenient to define a chain $\chain{g_k}{k}$ 
of functions in $D \contarrow D$ by stipulating
\begin{eqnarray*}
         g_0(D) & = & 0 \\[1mm]
         g_{k+1}(D) & = & 
\selectf{\neg{B}}(D) + g_k(\hpd{\omega}{v_1}{v'_1}(\selectf{B}(D))).
\end{eqnarray*}
Observe that for all $k$, and all $D$, we have 
\begin{equation}
\label{eq:gk-upper}
    \hpd{\omega_k}{v}{v'}(D) \leq g_k(D).
\end{equation}
This is trivial for $k = 0$, and for the inductive step we have
\begin{eqnarray*}
\hpd{\omega_{k+1}}{v}{v'}(D) & = &
\selectf{\neg{B}}(D) +
\hpd{\omega_k}{v}{v'}(\hpd{\omega_k}{v_1}{v'_1}(\selectf{B}(D)))
\\ & \leq &
\selectf{\neg{B}}(D) +
g_k(\hpd{\omega}{v_1}{v'_1}(\selectf{B}(D)))
\\ & = &
g_{k+1}(D).
\end{eqnarray*}
For all $k$, and all $D$, we also have
\begin{equation}
\label{eq:gk-lower}
g_k(D) \leq \limit{k}{\hpd{\omega_k}{v}{v'}(D)}.
\end{equation}
For $k = 0$ this is obvious, and for
the inductive step we have
\begin{eqnarray*}
g_{k+1}(D) & = & 
\selectf{\neg B}(D) +  g_k(\hpd{\omega}{v_1}{v'_1}(\selectf{B}(D)))
\\ & = &
\selectf{\neg B}(D) +  g_k(\limit{k}{\hpd{\omega_k}{v_1}{v'_1}(\selectf{B}(D))})
\\ & \leq &
\selectf{\neg B}(D) +  \limit{k}{\hpd{\omega_k}{v}{v'}(\limit{k}{\hpd{\omega_k}{v_1}{v'_1}(\selectf{B}(D))})}
\\ & = &
\selectf{\neg B}(D) +  \limit{k}{\hpd{\omega_k}{v}{v'}(\hpd{\omega_k}{v_1}{v'_1}(\selectf{B}(D)))}
\\ & = &
\limit{k}{(\selectf{\neg B}(D) + \hpd{\omega_k}{v}{v'}(\hpd{\omega_k}{v_1}{v'_1}(\selectf{B}(D))))}
\\ & = &
\limit{k}{\hpd{\omega_{k+1}}{v}{v'}(D)}.
\end{eqnarray*}
From (\ref{eq:gk-upper}) and (\ref{eq:gk-lower})
we see that for all $k$, and all $D$, we have
\[
\hpd{\omega_k}{v}{v'}(D) \leq g_k(D) \leq \limit{k}{\hpd{\omega_k}{v}{v'}(D)}
\]
from which we infer that
\begin{equation}
\label{eq:consistent1}
\mbox{for all } D:
\limit{k}{g_k(D)} = \limit{k}{\hpd{\omega_k}{v}{v'}(D)}.
\end{equation}
The virtue of working on $g_k$ rather than on $\omega_k$ is that
we can then prove the following result:
\begin{equation}
\label{eq:consistent2}
\mbox{for all $k$ and $D$},\ 
\sum_{s \in \fulls}F_k(s) \cdot D(s) = \sum_{s \in \fulls}F'(s) \cdot g_k(D)(s).
\end{equation}
where in Figure~\ref{fig:struct-sem} we defined $F_k$ as follows:
\begin{eqnarray*}
F_0(s) & = & 0 \\[1mm]
F_{k+1}(s) & = & \semc{S_1}(F_k)(s) \mbox{ if } \seme{B}s \\[1mm]
F_{k+1}(s) & = & F'(s) \mbox{ otherwise.}
\end{eqnarray*}
The proof of (\ref{eq:consistent2})
is by induction in $k$, where the base case $k = 0$ is obvious
as $F_0 = 0 = g_0(D)$.
For the inductive step, we have the calculation
\begin{eqnarray*}
& & \sum_{s \in \fulls}F_{k+1}(s) \cdot D(s)  \\
& = &
\sum_{s \in \fulls\ \mid\ \seme{B}s} \semc{S_1}(F_k)(s) \cdot D(s)
  +
\sum_{s \in \fulls\ \mid\ \seme{\neg B}s} F'(s) \cdot D(s)
\\ & = &
\sum_{s \in \fulls} \semc{S_1}(F_k)(s) \cdot \selectf{B}(D)(s)
  +
\sum_{s \in \fulls} F'(s) \cdot \selectf{\neg B}(D)(s)
\\ & = &
\sum_{s \in \fulls} F_k(s) \cdot \hpd{\omega}{v_1}{v'_1}(\selectf{B}(D))(s)
  +
\sum_{s \in \fulls} F'(s) \cdot \selectf{\neg B}(D)(s)
\\ & = &
\sum_{s \in \fulls} F'(s) \cdot g_k(\hpd{\omega}{v_1}{v'_1}(\selectf{B}(D)))(s)
  +
\sum_{s \in \fulls} F'(s) \cdot \selectf{\neg B}(D)(s)
\\ & = &
\sum_{s \in \fulls} F'(s) \cdot \left(g_k(\hpd{\omega}{v_1}{v'_1}(\selectf{B}(D))) + 
\selectf{\neg B}(D)\right)(s)
\\ & = &
\sum_{s \in \fulls} F'(s) \cdot g_{k+1}(D)(s) 
\end{eqnarray*}
where the 3rd equality comes from applying the outer (structural) induction hypothesis
to $S_1$ while the 4th equality comes from the inner induction
hypothesis (in $k$).

Since $F(s) = \limit{k}{F_k(s)}$ (by Figure~\ref{fig:struct-sem}),
the desired claim now follows from
\begin{eqnarray*}
\sum_{s \in \fulls} F(s) D(s)
& = & 
\sum_{s \in \fulls} \limit{k}{F_k(s) D(s)} \\
(\mbox{Lemma \ref{lem:sumlim-limsum}}) & = &
\limit{k}{\sum_{s \in \fulls} F_k(s) D(s)} \\
(\mbox{by (\ref{eq:consistent2})}) & = &
\limit{k}{\sum_{s \in \fulls} F'(s) g_k(D)(s)} \\
(\mbox{Lemma \ref{lem:sumlim-limsum}}) & = &
\sum_{s \in \fulls} \limit{k}{F'(s) g_k(D)(s)} \\
(\mbox{by (\ref{eq:consistent1})}) & = &
\sum_{s \in \fulls} F'(s) \limit{k}{\hpd{\omega_k}{v}{v'}(D)} \\
& = & \sum_{s \in \fulls} F'(s) D'(s).
\end{eqnarray*}

\end{itemize}

\section{Conclusion and Related Work}
\label{sec:conclude}
We have considered control flow graphs for probabilistic
imperative programs,
and developed a fixed-point based operational semantics
(which in a companion paper~\cite{Amt+Ban:ProbSlicing-2017}
we have used to reason about the correctness of
slicing such graphs).

We have stated and proved an adequacy result that
shows that for control flow graphs that are translations
of programs in a structured probabilistic language, our semantics is suitably
related to that language's denotational semantics
as presented in Gordon~\etal~\cite{Gor+etal:ICSE-2014} 
and Hur~\etal~\cite{Hur+etal:PLDI-2014}
(augmenting, in particular to handle conditioning,
early work by Kozen~\cite{Kozen:JCSS-81}).
In future work on probabilistic imperative programs, one thus has
the freedom to choose the semantics that best fits the given purpose.

\bibliographystyle{abbrv} 
\bibliography{bib1}

\begin{thebibliography}{10}

\bibitem{Aho:86:Dragon}
A.~V. Aho, R.~Sethi, and J.~D. Ullman.
\newblock {\em Compilers. Principles, Techniques and Tools}.
\newblock Addison-Wesley, 1986.

\bibitem{Amt+Ban:FoSSaCS-2016}
T.~Amtoft and A.~Banerjee.
\newblock A theory of slicing for probabilistic control flow graphs.
\newblock In B.~Jacobs and C.~L{\"{o}}ding, editors, {\em Foundations of
  Software Science and Computation Structures - 19th International Conference
  (FoSSaCS)}, volume 9634 of {\em Lecture Notes in Computer Science}, pages
  180--196. Springer-Verlag, 2016.

\bibitem{Amt+Ban:ProbSlicing-2017}
T.~Amtoft and A.~Banerjee.
\newblock A theory of slicing for probabilistic control-flow graphs.
\newblock Submitted for publication. An extended version of
  \cite{Amt+Ban:FoSSaCS-2016}, 2017.

\bibitem{Bal+Hor:AAD-1993}
T.~Ball and S.~Horwitz.
\newblock Slicing programs with arbitrary control flow.
\newblock In P.~Fritzon, editor, {\em Proceedings of the First International
  Workshop on Automated and Algorithmic Debugging (AADEBUG '93)}, volume 749 of
  {\em LNCS}, pages 206--222, London, UK, 1993. Springer-Verlag.

\bibitem{Gor+etal:ICSE-2014}
A.~D. Gordon, T.~A. Henzinger, A.~V. Nori, and S.~K. Rajamani.
\newblock Probabilistic programming.
\newblock In M.~B. Dwyer and J.~Herbsleb, editors, {\em ICSE, Future of
  Software Engineering track}, FOSE 2014, pages 167--181, New York, NY, USA,
  2014. ACM.

\bibitem{Hur+etal:PLDI-2014}
C.-K. Hur, A.~V. Nori, S.~K. Rajamani, and S.~Samuel.
\newblock Slicing probabilistic programs.
\newblock In K.~Pingali, editor, {\em Proceedings of the 35th ACM SIGPLAN
  Conference on Programming Language Design and Implementation}, PLDI '14,
  pages 133--144, New York, NY, USA, 2014. ACM.

\bibitem{Kozen:JCSS-81}
D.~Kozen.
\newblock Semantics of probabilistic programs.
\newblock {\em Journal of Computer and System Sciences}, 22:328--350, 1981.

\bibitem{Pod+Cla:TSE-1990}
A.~Podgurski and L.~A. Clarke.
\newblock A formal model of program dependences and its implications for
  software testing, debugging, and maintenance.
\newblock {\em {IEEE} Transactions on Software Engineering}, 16(9):965--979,
  Sept. 1990.

\bibitem{Schmidt:DenSemantics}
D.~A. Schmidt.
\newblock {\em Denotational Semantics, a Methodology for Language Development}.
\newblock Allyn and Bacon, Boston, 1986.

\bibitem{Winskel:semantics}
G.~Winskel.
\newblock {\em The Formal Semantics of Programming Languages}.
\newblock MIT Press, 1993.

\end{thebibliography}

\end{document}